\documentclass[journal,10pt]{IEEEtran}
\usepackage{xcolor}

\usepackage{eso-pic}
\usepackage{amsmath,amssymb,amsthm,amsfonts,bm}
\usepackage{graphicx}
\usepackage{algorithm}
\usepackage[noend]{algpseudocode}
\usepackage{booktabs}
\usepackage{placeins}
\usepackage{cite}
\usepackage{xcolor}
\usepackage{stmaryrd}
\usepackage{url}
\usepackage{tikz}
\usetikzlibrary{arrows.meta,positioning,decorations.pathreplacing,patterns,shapes.geometric}
\usepackage{url}
\usepackage[colorlinks=true,linkcolor=blue,citecolor=blue,urlcolor=blue]{hyperref}
\newcommand{\runin}[1]{\textbf{#1}}

\graphicspath{{figs/}}

\theoremstyle{plain}
\newtheorem{theorem}{Theorem}

\newtheorem{corollary}{Corollary}
\theoremstyle{definition}
\newtheorem{remark}{Remark}

\newcommand{\bw}{\mathbf{w}}
\newcommand{\ba}{\mathbf{a}}
\newcommand{\bp}{\mathbf{p}}
\newcommand{\bq}{\mathbf{q}}
\newcommand{\bC}{\mathbf{C}}
\newcommand{\bI}{\mathbf{I}}
\newcommand{\bR}{\mathbf{R}}
\newcommand{\bv}{\mathbf{v}}
\newcommand{\bef}{\mathbf{f}}
\newcommand{\bg}{\mathbf{g}}
\newcommand{\bd}{\mathbf{d}}
\newcommand{\bz}{\mathbf{z}}
\newcommand{\bM}{\mathbf{M}}

\newcommand{\Rg}{\bR_{\!g}}
\newcommand{\Reff}{\bR_{\!\text{eff}}}
\newcommand{\Cset}{\mathbb{C}}
\newcommand{\Mtot}{M_{\rm tot}}
\newcommand{\Rset}{\mathbb{R}}
\newcommand{\Tset}{\mathbb{T}}
\newcommand{\Gset}{\mathcal{G}}

\newcommand{\Mset}{\mathcal{M}}
\newcommand{\bA}{\mathbf{A}}

\newcommand{\bD}{\mathbf{D}}
\newcommand{\bx}{\mathbf{x}}

\newcommand{\bu}{\mathbf{u}}
\newcommand{\bB}{\mathbf{B}}
\newcommand{\bDelta}{\boldsymbol{\Delta}}
\newcommand{\E}{\mathbb{E}}
\DeclareMathOperator*{\argmax}{arg\,max}

\DeclareMathOperator{\re}{Re}
\DeclareMathOperator{\im}{Im}
\DeclareMathOperator{\tr}{tr}

\DeclareMathOperator{\Retr}{Retr}
\AtBeginDocument{
  \setlength{\abovedisplayskip}{4pt plus 1pt minus 2pt}
  \setlength{\belowdisplayskip}{4pt plus 1pt minus 2pt}
  \setlength{\abovedisplayshortskip}{2pt plus 1pt minus 1pt}
  \setlength{\belowdisplayshortskip}{2pt plus 1pt minus 1pt}
  \setlength{\textfloatsep}{5pt plus 2pt minus 2pt} % 图片/表格与正文之间的上下间距 (Default is often 20pt!)
  \setlength{\floatsep}{4pt plus 2pt minus 2pt}     % 两个相邻图表之间的间距
  \setlength{\intextsep}{5pt plus 2pt minus 2pt}    % 强制插入正文中间的图表间距 (e.g., [h])
}
\makeatletter
\def\section{\@startsection{section}{1}{\z@}%
  {1.5ex plus 1.0ex minus 1.0ex}% <- 缩小这行: 标题上方的空隙 (Space above)
  {0.4ex plus 0.5ex minus 0ex}%  <- 缩小这行: 标题下方的空隙 (Space below)
  {\normalfont\normalsize\centering\scshape}}

\def\subsection{\@startsection{subsection}{2}{\z@}%
  {1.5ex plus 1.0ex minus 1.0ex}% <- 缩小这行: 标题上方的空隙
  {0.4ex plus 0.5ex minus 0ex}%  <- 缩小这行: 标题下方的空隙
  {\normalfont\normalsize\itshape}}
\makeatother
\makeatletter
\renewenvironment{IEEEproof}[1][\IEEEproofname]{%
  \par\pushQED{\qed}\normalfont
  \trivlist
  \item[\hskip\labelsep
        \scshape #1\@addpunct{:}]\ignorespaces
}{%
  \popQED\endtrivlist\@endpefalse
}
\makeatother

\begin{document}

\AddToShipoutPictureFG*{%
  \AtPageUpperLeft{%
    \raisebox{-0.30in}[0pt][0pt]{%
      \makebox[\paperwidth][c]{%
        \fontsize{7}{8}\selectfont
        \textit{This work has been accepted to publish on the IEEE Transactions on Wireless Communications. Copyright transferred to IEEE.}%
      }%
    }%
  }%
}
% \title{A Unified Dual-Analysis Framework for Sparse-Array \\
% Near-Field Beam Focusing under Spatial \\
% Interference Suppression}

\title{Spatially Constrained Near-Field Beam Focusing for Sparse LEO
Satellite Formations: A Unified Dual Framework and Geometric SIR Limits}

\author{
\IEEEauthorblockN{Changhao He,
                  Xiaojuan Zhang,~\IEEEmembership{Senior Member,~IEEE,} Francois Chin Po Shin}
%\thanks{Manuscript received \today; revised XX. This work was supported in part by \ldots.}
 \thanks{C. H. He is with King Abdullah University of Science and Technology, Saudi Arabia (changhao.he@kaust.edu.sa).
 X. J. Zhang and C. P. Shin are with the Institute of Advanced Intelligence and Computing (IAIC), Agency for Science, Technology and Research, Singapore (xiaojuanzhang@ieee.org, Francois\_Chin@a-star.edu.sg).
 Corresponding author: X. J. Zhang. }
 \thanks{This work was supported by the Agency for Science, Technology and Research (A*STAR), Singapore under the Semiconductor RIE Flagship Programme (Integrated Circuit Design \& Systems), grant No. H26-MSE1535.}
}

\maketitle
% \begin{abstract}
% We study sparse-array near-field beam focusing with spatial interference
% suppression in coherent satellite formations. \rev{Existing methods address this
% problem} numerically by second-order cone programming (SOCP) with cutting-plane
% refinement, yet the achievable signal-to-interference ratio (SIR) and its link
% to classical adaptive beamforming lack an analytical characterization\rev{, which we
% provide through a Lagrangian-dual analysis}. First, every optimal beamformer is a
% generalized matched filter against an effective spatial covariance induced by an optimal dual measure. \rev{In this form, minimum-variance distortionless-response (MVDR) is the uniform-measure case, linearly-constrained minimum-variance (LCMV)
% a limiting case with a uniform component and discrete atoms, and cutting-plane SOCP a numerical realization of an optimal measure.}
% Second, \rev{an optimal dual measure can be chosen with at most $2M-1$ support
% points, where $M$ is the number of array weights.} Third, a closed-form upper bound on the mean SIR
% admits an asymptotic \rev{scaling law that is linear in $M$, i.e., logarithmic in
% decibels,} under near-collinear
% geometry\rev{. Numerical results show that several distinct designs, including a
% constant-modulus design obtained by Riemannian conjugate-gradient optimization,
% approach the bound at the tested geometry.}
% \end{abstract}

\begin{abstract}
We study sparse-array near-field beam focusing with spatial interference
suppression in coherent satellite formations. Existing methods address this
problem numerically using second-order cone programming (SOCP) with cutting-plane
refinement, but the achievable signal-to-interference ratio (SIR) and its
connection to classical adaptive beamforming remain analytically
uncharacterized. We address this gap through a Lagrangian-dual analysis.
We show that every optimal beamformer has the structure of a generalized
matched filter against an effective spatial covariance induced by an optimal
dual measure. This structure connects minimum-variance
distortionless-response (MVDR) beamforming as the uniform-measure case,
linearly-constrained minimum-variance (LCMV) beamforming as a limiting case
with a uniform component and discrete atoms, and cutting-plane SOCP as a
numerical realization of an optimal measure.
We further show that an optimal dual measure can be chosen with at most
$2M-1$ support points, where $M$ is the number of array weights.
A closed-form upper bound on the mean SIR yields an asymptotic
scaling law that is linear in $M$, or logarithmic in decibels,
under near-collinear geometry. Numerical results show that several
distinct designs, including a constant-modulus design obtained by
Riemannian conjugate-gradient optimization, approach this bound at the
tested geometry.
\end{abstract}

\begin{IEEEkeywords}
Near-field beam focusing, spatial interference suppression, LEO satellite
formation, sparse arrays, semi-infinite programming, Lagrangian duality,
constant-modulus beamforming, manifold optimization.
\end{IEEEkeywords}

\section{Introduction}\label{sec:intro}

% \IEEEPARstart{T}{he} migration toward regenerative-payload non-terrestrial
% networks (NTN) lets low-Earth-orbit (LEO) satellites act as on-orbit base
% stations rather than transparent relays~\cite{ref:3gpp-ntn-r19}. A downlink
% toward an airborne terminal then shares spectrum with terrestrial users
% beneath the footprint, and co-channel emissions at ground level can exceed the
% interference margins of the underlying cell. The mitigation is to focus
% radiated energy at the airborne target while constraining the field over a designated terrestrial protection region. Realized through a coherent multi-satellite formation, this is the \emph{spatially constrained near-field beam focusing} problem studied here. \rev{The term refers throughout to interference constraints imposed over a continuous region rather than at a finite set of discrete interferer locations.} That structure is not specific to satellites. Distributed coherent elements focusing a converging wavefront at a target, while attenuating the field over a continuum of constrained locations, also arise in ultra-large-scale arrays~\cite{ref:nf-survey} and distributed \rev{multiple-input multiple-output (MIMO)} systems~\cite{ref:dmimo}. The analysis below is stated for a general distributed array and instantiated on the LEO formation.

\IEEEPARstart{T}{he} migration toward regenerative-payload non-terrestrial
networks (NTN) allows low-Earth-orbit (LEO) satellites to operate as
on-orbit base stations rather than transparent relays~\cite{ref:3gpp-ntn-r19}.
A downlink to an airborne terminal may share spectrum with terrestrial
users beneath the satellite footprint, where co-channel emissions can
exceed the interference limits of terrestrial cells. This motivates
focusing radiated energy at the airborne target while limiting the field
over a designated terrestrial protection region. With a coherent
multi-satellite formation, this becomes the \emph{spatially constrained
near-field beam focusing} problem considered in this paper.
Here, spatial constraints are imposed over a continuous region
rather than at a finite set of discrete interferer locations.
The same problem arises beyond satellite systems. Distributed coherent
elements that focus energy at a target while suppressing interference
over a continuous region also appear in ultra-large-scale
arrays~\cite{ref:nf-survey} and distributed
multiple-input multiple-output (MIMO)
systems~\cite{ref:dmimo}. The following analysis applies to a general
distributed array, with the LEO formation serving as the specific
application.

% \subsection{Problem Setting and Challenges}

% This formulation gives rise to three principal challenges.

% This formulation gives rise to three principal challenges.
% % summarized here and elaborated quantitatively in Section~\ref{sec:model}.

% These three challenges jointly motivate the unified dual framework
% developed in this paper.
% \subsection{Three Principal Challenges}
% Joint energy focusing and spatial protection pose three key challenges, addressed herein.
% \rev{First, a} continuous protection region produces
% infinitely many constraints, demanding cutting-plane
% methods~\cite{ref:kelley1960,ref:hettich1993} whose convergence and
% inter-sample leakage need careful analysis. \rev{Second, because}
% the target altitude is far smaller than the orbital altitude, \rev{every
% platform sees target and interferers from almost the same direction}, imposing an
% algorithm-independent ceiling on the signal-to-interference ratio (SIR).
% \rev{Third, NTN power amplifiers held near saturation for efficiency favor
% phase-only operation, and designing against} a
% per-element constant-modulus constraint \rev{replaces} closed-form linear
% solutions with non-convex manifold optimization.

\subsection{Three Principal Challenges}
Joint energy focusing and spatial protection present three main challenges.
First, a continuous protection region introduces
infinitely many constraints, motivating cutting-plane
methods~\cite{ref:kelley1960,ref:hettich1993} whose convergence and
inter-sample leakage require careful analysis. Second, because
the target altitude is much smaller than the orbital altitude, each
platform views the target and interferers from nearly the same direction,
placing an algorithm-independent ceiling on the signal-to-interference
ratio (SIR). Third, NTN power amplifiers operated near saturation
for efficiency favor phase-only operation, while imposing
a per-element constant-modulus constraint leads from
closed-form linear solutions to non-convex manifold optimization.

% \subsection{Related Work and Motivation}

% These observations motivate the development of a unified analytical framework
% that connects convex beamforming, adaptive array processing, and
% constant-modulus implementation under a common dual interpretation.
% \subsection{Related Work and Motivation}

\subsection{Related Work and Motivation}

% \rev{The problem connects three research areas that developed independently. Satellite precoding supplies the link and payload model but treats the array as one co-located aperture. Near-field focusing supplies the spherical-wave mechanism but assumes dense, electrically large apertures. Semi-infinite and adaptive-array optimization supplies the machinery for continuous constraints but is rarely tied to a geometric limit on what any beamformer can achieve. A sparse, kilometer-scale formation protecting a
% continuous ground region falls between the three.}

The problem draws on three areas of research. Satellite precoding
provides the link and payload models, but typically considers a single
co-located aperture. Near-field focusing accounts for spherical-wave
propagation, mainly in dense, electrically large arrays. Semi-infinite
and adaptive-array optimization provides tools for handling continuous
constraints, but has rarely been connected to geometric performance
limits. Sparse, kilometer-scale satellite formations with interference
constraints over a continuous ground region have received comparatively
little attention in these areas.

% Early GEO high-throughput satellite studies developed gateway-side linear and multicast precoders under fixed-payload assumptions~\cite{ref:vazquez2016,ref:joroughi2017,ref:perez2019}. Subsequent LEO NTN work considered statistical \rev{channel state information (CSI)}, hybrid architectures, and hardware impairments~\cite{ref:you2022hybrid,ref:you2022twin,ref:chen2024survey}. A parallel line treats constellations as
% distributed antenna arrays: cell-free and distributed massive-MIMO
% cooperation~\cite{ref:abdelsadek2022dmimo,ref:abdelsadek2021cf}, user-centric
% clustering for asynchronous transmission~\cite{ref:humadi2024dynamic}, and
% experimental multi-satellite \rev{multi-user MIMO (MU-MIMO)} with carrier-phase
% synchronization~\cite{ref:storek2020testbed} or wideband swarm
% beamforming~\cite{ref:merlano2024swarm}. Both optimize capacity toward discrete
% users; neither \rev{characterizes the achievable performance} when the interference specification is a
% continuous region.

Early GEO high-throughput satellite studies developed gateway-side linear
and multicast precoders under fixed-payload
assumptions~\cite{ref:vazquez2016,ref:joroughi2017,ref:perez2019}.
Subsequent work on LEO NTN considered statistical
channel state information (CSI), hybrid architectures, and hardware
impairments~\cite{ref:you2022hybrid,ref:you2022twin,ref:chen2024survey}.
Another line of research treats constellations as distributed antenna
arrays, including cell-free and distributed massive-MIMO
cooperation~\cite{ref:abdelsadek2022dmimo,ref:abdelsadek2021cf},
user-centric clustering for asynchronous
transmission~\cite{ref:humadi2024dynamic}, and experimental
multi-satellite multi-user MIMO (MU-MIMO) with carrier-phase
synchronization~\cite{ref:storek2020testbed} or wideband swarm
beamforming~\cite{ref:merlano2024swarm}. Both lines of research focus
on capacity optimization for discrete users, without
characterizing the achievable performance under interference
constraints imposed over a continuous region.

% Spherical-wave propagation enables angle--range focusing with extra-large arrays beyond far-field beamforming~\cite{ref:nf-survey,ref:lu2022xlarray,ref:you2024nfbm}. Extensions include wideband and sensing--communication systems~\cite{ref:wang2024beamfocus}, modular arrays~\cite{ref:li2022modular}, and location- and user-domain beamforming~\cite{ref:wu2023ldma,ref:zhang2022beamfocus}. In each case the aperture
% is dense and co-located and the protected set is discrete. \rev{In a sparse
% kilometer-scale formation the near-field regime instead arises in spite of the
% long propagation distance, the target falling
% well inside the formation's Fraunhofer distance.} Whether the focusing gain
% survives that geometry, and what SIR limit follows, is unresolved.

Spherical-wave propagation enables angle--range focusing with extra-large
arrays beyond conventional far-field
beamforming~\cite{ref:nf-survey,ref:lu2022xlarray,ref:you2024nfbm}.
Extensions include wideband and sensing--communication
systems~\cite{ref:wang2024beamfocus}, modular
arrays~\cite{ref:li2022modular}, and location- and user-domain
beamforming~\cite{ref:wu2023ldma,ref:zhang2022beamfocus}.
These studies generally consider dense, co-located apertures and
discrete protected locations. For a sparse, kilometer-scale
formation, near-field effects remain significant despite the long
propagation distance, as the target lies well within the formation's
Fraunhofer distance. Whether such a formation can retain the focusing
gain, and what SIR limit it can achieve, remains unclear.

% Semi-infinite spatial constraints are commonly handled by cutting-plane
% \rev{second-order cone programming (SOCP)}~\cite{ref:kelley1960,ref:hettich1993}, which is numerically effective
% but gives limited insight into the structure of optimal solutions or the
% achievable SIR. Classical adaptive beamforming offers closed-form
% alternatives, \rev{minimum-variance distortionless-response (MVDR) and
% linearly-constrained minimum-variance (LCMV)}~\cite{ref:capon1969,ref:frost1972}, with robust
% extensions for CSI uncertainty~\cite{ref:vorobyov2003}, yet their
% relationship to semi-infinite SOCP in near-field focusing is unclear.
% Riemannian optimization on the unit torus is standard for constant-modulus
% beam synthesis~\cite{ref:absil2008,ref:sun2017}, but its loss relative to
% unconstrained beamforming in sparse near-field arrays is not established.

Semi-infinite spatial constraints are commonly handled using cutting-plane
second-order cone programming (SOCP)~\cite{ref:kelley1960,ref:hettich1993}.
Although numerically effective, this approach provides limited insight into
the structure of optimal solutions or the achievable SIR. Classical adaptive
beamforming offers closed-form alternatives, including
minimum-variance distortionless-response (MVDR) and
linearly-constrained minimum-variance (LCMV)~\cite{ref:capon1969,ref:frost1972},
with robust extensions for CSI uncertainty~\cite{ref:vorobyov2003}.
However, their relationship to semi-infinite SOCP in near-field focusing
remains unclear. Riemannian optimization on the unit torus is well established
for constant-modulus beam synthesis~\cite{ref:absil2008,ref:sun2017},
but its performance loss relative to unconstrained beamforming in sparse
near-field arrays has not been characterized.

% \rev{This paper joins the three strands in a single Lagrangian-dual
% framework that characterizes the
% optimization problem itself: a closed-form dual-measure representation of
% the optimal beamformer, a finitely atomic choice of that measure, and
% a geometry-dependent bound on the achievable mean SIR that holds for every
% feasible design. MVDR, LCMV and cutting-plane SOCP then appear as structural
% specializations of a single dual-measure template, and constant-modulus
% optimization as a hardware-constrained realization outside it.}

This paper brings these three areas together through a
Lagrangian-dual framework, in which MVDR, LCMV and cutting-plane SOCP
share a common dual-measure structure, while constant-modulus
optimization is treated separately as a hardware-constrained mean-SIR
benchmark.

\subsection{Contributions and Novelty}\label{sec:contribs}

This paper develops a unified Lagrangian-dual framework for
spatially constrained near-field beam focusing. We consider an
$M$-dimensional beamforming vector $\bw\in\Cset^M$, a target
steering vector $\ba_t$, and a continuous protection region $\Gset$
with steering-vector map $\ba(\cdot)$. The analysis establishes
the optimal beamformer structure, a finite atomic representation
of an optimal dual measure, and a geometry-dependent mean-SIR bound.
Contributions C1 and C2 require only that $\Gset$ be compact, $\ba(\cdot)$ be
continuous, and a strictly feasible design exist. Both apply to
any array geometry and propagation model. C2 further shows that
the dual measure in C1 can be chosen with finite support,
while C3 is independent of both results.

\begin{enumerate}

\item \runin{(C1) Optimal structural characterization.}
Every optimal linear beamformer $\bw^\star$ satisfies
$\Reff\bw^\star \propto \ba_t$, where $\Reff$ is an effective
spatial covariance generated by an optimal dual measure.
When $\Reff\succ0$, the solution admits the closed form
$\bw^\star \propto \Reff^{-1}\ba_t$
(Theorem~\ref{thm:structure}). The covariance is defined
in~\eqref{eq:Reff-def}, and the optimal dual measure is supported
on the active constraints in $\Gset$. This connects semi-infinite
SOCP and classical adaptive beamforming through a common
optimality structure.

\item \runin{(C2) Finite-dimensional dual certificate.}
An optimal dual measure can be chosen with at most $2M-1$ atoms,
a bound determined solely by the dimension $M$ of $\bw$ and
independent of discretization (Theorem~\ref{thm:sparsity}).
This gives a finite representation of an optimal dual measure
despite the continuum of constraints, but does not bound the
working set of any implementation
(Remark~\ref{rem:sparsity-use}).

\item \runin{(C3) Geometry-dependent performance limit.}
The formation and protection region determine the spatial-average
steering covariance $\Rg$ of~\eqref{eq:Rg-def}. The achievable
mean SIR is bounded by $\rho^\star=\ba_t^H\Rg^{-1}\ba_t$
(Theorem~\ref{thm:bound}), which also equals the classical Capon
output SIR~\cite{ref:capon1969} when $\Rg$ is taken as the
interference covariance. The bound follows from Cauchy--Schwarz
and applies to every design provided $\Rg\succ0$.
The associated scaling law (Theorem~\ref{thm:asymp}) is linear
in $M$, or logarithmic in decibels, for collinear platforms
over a ground segment in the near-collinear regime.
At the tested geometry, four structurally different designs
operate within $0.11$~dB of the bound, while a well-posed
LCMV design falls $4.5$~dB below it
(Section~\ref{sec:lcmv-failure}).
Thus, the geometry determines the performance ceiling,
while the beamforming design determines how closely it is
approached.

\end{enumerate}

We examine the dual structure and the bound with MVDR, LCMV,
cutting-plane SOCP, and Riemannian conjugate-gradient (CG)
optimization on the constant-modulus torus. The mean-SIR bounds of
the tested layouts with equal $M$ differ by less than $0.03$~dB, and
an illustrative link-level receive gain adds in decibels under the
stated receiver normalization, whereas a larger transmit-side
aperture yields diminishing gains.

\subsection{Notation and Outline}
Bold lowercase and uppercase letters denote vectors and matrices,
respectively. $\bw^H$, $\bw^*$, and $\bw^T$ denote the conjugate
transpose, complex conjugate, and transpose, respectively.
$\re\{\cdot\}$ and $\im\{\cdot\}$ denote the real and imaginary
parts. $\|\cdot\|_2$ denotes the Euclidean norm.
The symbols $\odot$ and $\oslash$ denote element-wise
multiplication and division, respectively, and $|\bw|$ denotes
the element-wise modulus $[|w_1|,\ldots,|w_M|]^T$.
$\sup_{\bq \in \Gset} f(\bq)$ denotes the pointwise
supremum of $f$ over $\Gset$.
$\Tset = \{z \in \Cset : |z|=1\}$ denotes the unit complex
circle. A matrix approximately proportional to the identity
is referred to as near-scalar.

% Section~\ref{sec:model} presents the system model and introduces
% the near-collinear regime associated with geometric degeneracy.
% Section~\ref{sec:dual} develops the Lagrangian-dual analysis
% (Theorems~\ref{thm:structure}--\ref{thm:asymp}).
% Section~\ref{sec:algos} describes the MVDR, LCMV, cutting-plane
% SOCP, and Riemannian constant-modulus designs.
% Section~\ref{sec:numerics} presents numerical results for a
% representative LEO formation and examines how different
% architectural choices affect the SIR ceiling.

Sections~\ref{sec:model}--\ref{sec:numerics} present the system
model, the dual analysis, the four designs, and the numerical
results, respectively, and Section~\ref{sec:conclusion} concludes
the paper.

\section{System Model}\label{sec:model}
\begin{figure*}[!t]
\centering
\includegraphics[width=0.95\textwidth]{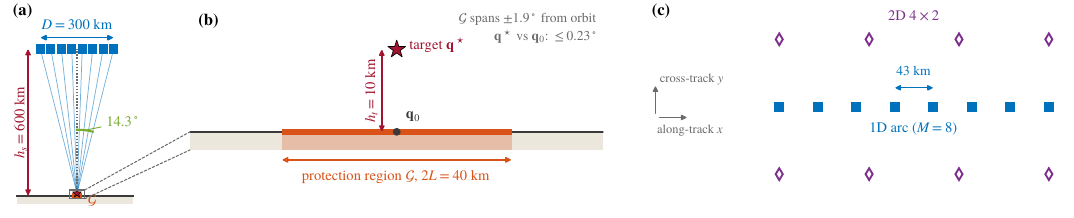}
\caption{System geometry. (a) Side view of the eight-platform formation and protected ground region, to scale. (b) Target and protection-region detail. (c) Top view, to scale, of the 1D arc and an illustrative $4\times2$ layout.}
\label{fig:system}
\end{figure*}

This section defines the array geometry, near-field channel, SIR metrics,
and reference SOCP formulation. We consider a coherent LEO formation
of $M$ satellites at altitude $h_s$ with baseline $D$, transmitting
jointly to an airborne target at altitude $h_t\ll h_s$ while
protecting a ground footprint of radius $L$ from coherent interference.

The reference values ($M=8$, $D=300$~km, $h_s=600$~km) specify
an idealized, straight, common-altitude formation rather than a deployed
circular-orbit constellation. A $300$~km along-track chord at
$600$~km altitude spans about $2.5^\circ$ of orbital arc, or $40$~s
of flight, with spacing $D/(M-1)\approx43$~km. This is a dedicated
formation-flying cluster rather than neighboring satellites of a
broadband shell, such as the $550$~km, $53^\circ$ Starlink shell
with $72$ planes of $22$ satellites~\cite{ref:fcc2148} and an
in-plane spacing of about $1{,}980$~km. Fig.~\ref{fig:bound-vs-MD}
varies $M$ and $D$, and Section~\ref{sec:nearfield} gives the
near-field operating range in $(f,h_s)$.

% \subsection{Geometry and Channel}
% The geometry is illustrated in Fig.~\ref{fig:system}: $M$ elements at
% positions $\bp_m \in \Rset^3$ with aperture span
% $D \triangleq \max_{m,n}\|\bp_m - \bp_n\|$ operate at wavelength $\lambda$
% toward a target $\bq^\star$ at altitude $h_t$.

% Spectrum coexistence in
% NTN downlinks \rev{requires protecting two qualitatively different groups of users}, and the
% dual framework treats them uniformly through the compact union
% \begin{equation}\label{eq:G-union}
% \Gset \;\triangleq\; \Gset_{\rm gnd} \,\cup\, \Gset_{\rm air}
% \;\subset\; \Rset^3 .
% \end{equation}
% Here $\Gset_{\rm gnd}$ is a bounded near-ground region enclosing terrestrial
% co-channel users with $\sup_{\bq\in\Gset_{\rm gnd}} q_z \le h_g \ll h_t$, and
% $\Gset_{\rm air}$ a bounded annular region enclosing other airborne users with
% $\sup_{\bq\in\Gset_{\rm air}}|q_z-h_t|\le\Delta_t$, lying outside the main-lobe
% disc of radius $r_{\rm ML}\triangleq\lambda (h_s-h_t)/D$ around $\bq^\star$. All
% structural results in Section~\ref{sec:dual} are stated for a generic compact
% $\Gset$ and apply to~\eqref{eq:G-union} verbatim. The near-collinear refinement of Theorem~\ref{thm:asymp} is specific to $\Gset_{\rm gnd}$\rev{. The
% airborne region} $\Gset_{\rm air}$ is separated
% from $\bq^\star$ by at least $\arctan\!\big(r_{\rm ML}/(h_s-h_t)\big) = \lambda/D$,
% a full beamwidth.

\subsection{Geometry and Channel}
Fig.~\ref{fig:system} shows $M$ elements at positions
$\bp_m \in \Rset^3$, with aperture span
$D \triangleq \max_{m,n}\|\bp_m - \bp_n\|$, operating at wavelength
$\lambda$ toward a target $\bq^\star$ at altitude $h_t$.

Spectrum coexistence in NTN downlinks requires protecting two
distinct user groups, represented by the compact union
\begin{equation}\label{eq:G-union}
\Gset \;\triangleq\; \Gset_{\rm gnd} \,\cup\, \Gset_{\rm air}
\;\subset\; \Rset^3 .
\end{equation}
Here $\Gset_{\rm gnd}$ is a bounded near-ground region containing terrestrial
co-channel users with $\sup_{\bq\in\Gset_{\rm gnd}} q_z \le h_g \ll h_t$, and
$\Gset_{\rm air}$ is a bounded annular region containing other airborne users
with $\sup_{\bq\in\Gset_{\rm air}}|q_z-h_t|\le\Delta_t$, outside the main-lobe
disc of radius $r_{\rm ML}\triangleq\lambda (h_s-h_t)/D$ around $\bq^\star$.
All structural results in Section~\ref{sec:dual} apply to any compact
$\Gset$, including~\eqref{eq:G-union}. The near-collinear refinement of
Theorem~\ref{thm:asymp} applies specifically to $\Gset_{\rm gnd}$.
On the target-altitude cross-section $q_z=h_t$, the boundary of
$\Gset_{\rm air}$ subtends $\arctan(r_{\rm ML}/(h_s-h_t))
=\arctan(\lambda/D)\approx\lambda/D$; for $|q_z-h_t|\le\Delta_t$,
the corresponding angular separation is altitude-dependent.

With wavenumber $k = 2\pi/\lambda$, the scalar near-field generated by
the weight vector $\bw = [w_1, \dots, w_M]^T \in \Cset^M$ at a field
point $\bq$ is
\begin{equation}\label{eq:field}
E(\bq) = \sum_{m=1}^{M} w_m \cdot \frac{e^{-jk\|\bq - \bp_m\|}}{\|\bq - \bp_m\|} = \ba(\bq)^H \bw,
\end{equation}
where $[\ba(\bq)]_m = e^{+jk r_m(\bq)}/r_m(\bq)$ and
$r_m(\bq) = \|\bq - \bp_m\|$, and
$\ba_t \triangleq \ba(\bq^\star)$ denotes the target steering vector.
The weight vector $\bw$ is the optimization variable. When platform $m$ carries $N_m$ co-located antennas, we adopt a
hierarchical model in which each platform combines its local elements
using a unit-norm target-matched weight
$\bw_m^{(\rm loc)} = \ba_m(\bq^\star)/\|\ba_m(\bq^\star)\|$,
while the inter-platform optimization operates on the resulting $M$
super-element outputs. This architecture is used in demonstrated
multi-platform coherent testbeds~\cite{ref:storek2020testbed,ref:merlano2024swarm}
and emerging long-baseline distributed arrays~\cite{ref:formation}. Let $\ba_m^{\rm loc}(\bq)$ be the local steering vector relative
to the platform phase center, excluding the center-to-field
propagation factor. Under target-matched local combining,
define the exact complex array factor
$c_m(\bq)\triangleq
[\ba_m^{\rm loc}(\bq^\star)]^H\ba_m^{\rm loc}(\bq)/
\|\ba_m^{\rm loc}(\bq^\star)\|$.
Model~\eqref{eq:field} then applies with
$[\ba(\bq)]_m\mapsto c_m(\bq)[\ba(\bq)]_m$.
Since $\|\bw_m^{(\rm loc)}\|_2=1$, platform $m$ radiates $|w_m|^2$,
so $\|\bw\|_2^2\le P_{\rm tot}$ also applies to the physical weights,
while the coherent local target response can scale as $\sqrt{N_m}$.
The magnitude-only approximation $c_m(\bq)\approx\sqrt{G_m(\bq)}$,
with $G_m(\bq)=|c_m(\bq)|^2$, neglects the residual local-array
phase, so the separable-gain interpretation in
Remark~\ref{rem:per-platform} is approximate, whereas
Theorem~\ref{thm:bound} holds for the exact complex steering vector.
The inter-platform analysis therefore uses $N_m=1$, and the
hierarchical and joint architectures are compared in
Section~\ref{sec:per-platform-numerics}.

As a restriction of fully joint processing, the hierarchical model
cannot exceed its optimum. The ceiling derived below therefore
applies to the realizable hierarchical architecture, not to full
$MN$-element coherence, and Sections~\ref{sec:degeneracy}
and~\ref{sec:per-platform-numerics} quantify the gap.

% \subsection{Geometric Regime and Constraint Structure}\label{sec:model-regime}

% \subsubsection{Near-collinear Regime}\label{sec:near-collinear}
% We call the regime in which $\eta \triangleq h_t/h_s$, $D/h_s$ and $L/h_s$
% are all small the \emph{near-collinear regime}. Its defining feature is an
% unavoidable angular degeneracy: viewed from the array, the sub-nadir
% projection $\bq_0$ of $\bq^\star$ onto the ground lies within
% \rev{$\mathcal{O}(\eta D/h_s)$ of $\bq^\star$ ($\le0.23^\circ$ here), a gap that
% remains small throughout this regime}. \rev{The ground protection region we consider contains a neighborhood of
% $\bq_0$, the closest ground point to $\bq^\star$;} \rev{no single platform resolves the two directions. The formation does, but averaging the leakage over $\Gset$ whitens its ground covariance, once the platforms are resolvable ($M/\Phi_{\min}\ll1$ in Theorem~\ref{thm:asymp}), and fixes an SIR ceiling set by $M$ and $\eta$ alone}, independent of $L$ and of the optimization algorithm. This
% degeneracy is what the asymptotic analysis below exploits.

\subsubsection{Near-Collinear Regime}\label{sec:near-collinear}
We refer to the regime where $\eta \triangleq h_t/h_s$, $D/h_s$ and $L/h_s$
are all small as the \emph{near-collinear regime}. Its main feature is an
unavoidable angular degeneracy: viewed from the array, the sub-nadir
projection $\bq_0$ of $\bq^\star$ onto the ground lies within
$\mathcal{O}(\eta D/h_s)$ of $\bq^\star$ ($\le0.23^\circ$ here), and this
separation remains small throughout the regime. The ground protection
region contains a neighborhood of $\bq_0$, the closest ground point to
$\bq^\star$; no single platform can resolve the two directions.
The formation can distinguish them, but averaging the leakage over $\Gset$
makes the ground covariance nearly scalar once the platforms are resolvable
($M/\Phi_{\min}\ll1$ in Theorem~\ref{thm:asymp}). This yields an SIR ceiling
determined by $M$ and $\eta$ alone, independent of $L$ and the optimization
algorithm. The asymptotic analysis below exploits this degeneracy.

% \subsubsection{Constant-Modulus Hardware Restriction}\label{sec:cm}

% Satellite transmitters are constrained by size, weight and power, and their
% solid-state amplifiers are typically operated near saturation for efficiency.
% \rev{Saturation alone does not force exact per-element constant modulus, but
% it makes constant-envelope, phase-only operation the practically relevant
% hardware regime and the conservative one to design against, since a payload
% retaining amplitude control can only do better. We therefore adopt
% $|w_m|=\sqrt{P_m}$, $P_m=P_{\rm tot}/M$, with freely adjustable phases.} This restricts $\bw$ to the \rev{torus $\Mset_{\rm CM}=(\sqrt{P_m}\,\Tset)^M$}, making the problem non-convex. Classical closed-form beamformers are then generally infeasible and convex relaxations often fail to
% recover rank-one solutions, so whether constant-modulus beamforming can
% approach the unconstrained performance is an open practical question and
% motivates the manifold algorithm of Section~\ref{sec:algos}.

\subsubsection{Constant-Modulus Hardware Restriction}\label{sec:cm}

Satellite transmitters face size, weight and power constraints, and their
solid-state amplifiers typically operate near saturation for efficiency.
Although saturation does not require exact per-element constant modulus,
it favors constant-envelope, phase-only operation. We adopt this as a
conservative hardware constraint, since a payload with amplitude control
can achieve at least the same performance. Specifically, we impose
$|w_m|=\sqrt{P_m}$, $P_m=P_{\rm tot}/M$, with freely adjustable phases.
This restricts $\bw$ to the torus
$\Mset_{\rm CM}=(\sqrt{P_m}\,\Tset)^M$, making the problem non-convex.
Classical closed-form beamformers are then generally infeasible, and convex
relaxations often fail to recover rank-one solutions. Whether
constant-modulus beamforming can approach unconstrained performance
therefore remains a practical question and motivates the manifold
algorithm in Section~\ref{sec:algos}.

% \subsection{\rev{Near-Field Operating Condition}}\label{sec:nearfield}
% \begin{revblock}
% The spherical-wave model~\eqref{eq:field} is a requirement of the operating
% point rather than a modeling preference. We state the criterion that decides
% the question, evaluate it here, and give the range over which the analysis
% applies.

% The plane-wave model is valid for an aperture of extent $D$ at range $R$
% when $R \ge 2D^2/\lambda$, equivalently when the residual quadratic
% phase $\Delta\phi \approx k D^2/(8R)$ is a small fraction of a radian; setting
% $\Delta\phi = \pi/8$ defines a \emph{critical baseline}
% \begin{equation}\label{eq:Dcrit}
% D_{\rm crit} \triangleq \sqrt{\lambda R/2},
% \end{equation}
% below which a formation behaves as a far-field array.

\subsection{Near-Field Operating Condition}\label{sec:nearfield}
The spherical-wave model~\eqref{eq:field} is required under the operating
conditions considered here, rather than adopted as a modeling preference.

A conventional Fraunhofer phase-error criterion for the aperture
transverse to the line of sight, $D_\perp$, is
$R\ge 2D_\perp^2/\lambda$, corresponding to
$\Delta\phi\approx kD_\perp^2/(8R)\le\pi/8$.
It is an engineering threshold, not a geometry-independent
necessary-and-sufficient plane-wave condition. Define the
critical transverse aperture
\begin{equation}\label{eq:Dcrit}
D_{\rm crit} \triangleq \sqrt{\lambda R/2},
\end{equation}
associated with that conventional phase-error threshold.

At $f=2$~GHz, the near-broadside reference geometry has
$D_\perp\approx D=300$~km and $D_{\rm crit}=214$~m, with a
$7.7\times10^5$-rad quadratic phase neglected by a plane-wave
model. Over $f=1$--$30$~GHz and $h_s=400$--$1200$~km,
$D_{\rm crit}=46$--$424$~m. For a horizontal baseline viewed at depression angle $\alpha$, $D_\perp=D\sin\alpha$.
For a spherical Earth of radius $6371$~km and $h_t=10$~km, the
limiting depression angle is $23.7^\circ$ at $h_s=600$~km and
$19.5^\circ$ at $h_s=400$~km. Across the stated ranges, the
conservative ratio $D_\perp/D_{\rm crit}$ still exceeds
$1.7\times10^2$, well within the near-field regime.

Two length scales must be distinguished. The formation is sparse
because $D/(M-1)\approx43$~km greatly exceeds $\lambda$, producing
a densely aliased array factor, and near-field because
$D\gg D_{\rm crit}$. This sparsity differs from sparse sampling of
an electrically contiguous extremely large array.

For $D<D_{\rm crit}$, the steering vector $\ba(\bq)$ of a
uniform linear array approaches the Vandermonde form
$[1,z(\bq),\dots,z(\bq)^{M-1}]^T$, $|z(\bq)|=1$,
up to a range-dependent amplitude. Theorems~\ref{thm:structure}--\ref{thm:bound}
still apply because they require only continuity of $\ba(\cdot)$,
strict feasibility, and the spanning condition. Thus, the
framework includes the far field, although its approximation
is inaccurate at the present operating point.

At $D=300$~km the plane-wave target coherence is $0.975$,
yet across $400$ nearby baselines its median drops to $0.43$,
causing a $7.2$~dB median matched-filter loss. A baseline
sweep confirms that the models agree within $0.0025$~dB
below $D_{\rm crit}/2$; their computed $\rho^\star$ values
differ by at most $0.004$~dB.

% \subsection{\rev{Formation Architecture and Coherence Assumptions}}\label{sec:architecture}
% \begin{revblock}
% The formation, the airborne target and the terrestrial
% protection region are connected by unobstructed line-of-sight paths
% and~\eqref{eq:field} carries no diffuse component, as is standard for a
% satellite-to-air link above the troposphere and for the specular satellite-to-ground path at $2$~GHz. A diffuse component would add a platform-uncorrelated term to the ground covariance: Theorem~\ref{thm:bound}
% still holds; Theorem~\ref{thm:structure} needs a stochastic extension.

\subsection{Formation Architecture and Coherence Assumptions}\label{sec:architecture}
We assume unobstructed line-of-sight paths between the formation,
the airborne target, and the terrestrial protection region.
The channel model in~\eqref{eq:field} neglects diffuse propagation,
consistent with the satellite-to-air link above the troposphere and
the specular satellite-to-ground path at $2$~GHz. A diffuse component
would introduce a platform-uncorrelated term into the ground covariance.
Theorem~\ref{thm:bound} would still apply, whereas
Theorem~\ref{thm:structure} would require a stochastic extension.

% The channel is deterministic given geometry, so no per-element channel estimate is needed. The ground covariance and $\ba_t$ follow from the platform positions $\{\bp_m\}$, the target position $\bq^\star$ and $\Gset$, which must be known to the accuracy quantified below. The beamformer is recomputed at the rate the geometry changes. CSI acquisition is thereby replaced by the
% positioning-and-synchronization problem quantified next. Each platform combines its $N_m$ elements with a fixed boresight weight toward $\bq^\star$, and the $M$ outputs are combined digitally, so $\bw \in \Cset^M$ is the only design variable. Against fully joint processing over all $\Mtot=MN$
% elements this hierarchy needs one radio-frequency (RF) chain per platform instead of $N$, an
% inter-platform rather than a network-wide per-element phase reference, and backhaul carrying weights rather than element samples. Both require the same information $\{\bp_m\},\bq^\star,\Gset$, but the hierarchical ceiling is set by $M$
% (Fig.~\ref{fig:arch}) and the joint one by the number of retained covariance modes, which we define as the number of eigenvalues of the design's ground covariance within a stated tolerance of the largest. It is $M$ for the near-scalar hierarchical design at any tolerance and grows only slowly in $N$ under joint processing, from $8$ at $N=1$ to $48$ at $N=729$ (Section~\ref{sec:per-platform-numerics}).

The channel is determined by geometry, so per-element channel estimation
is unnecessary. The ground covariance and $\ba_t$ are computed from
$\{\bp_m\}$, $\bq^\star$ and $\Gset$, which must be known to the
accuracy specified below, and the beamformer is updated as the geometry
changes. CSI acquisition is thus replaced by positioning and
synchronization requirements.
Compared with fully joint processing over all $\Mtot=MN$ elements, the
hierarchical architecture requires one radio-frequency (RF) chain per
platform rather than up to $N$, an inter-platform phase reference rather than
network-wide per-element phase synchronization, and backhaul for weights
rather than element samples. Its mean-SIR ceiling is governed by its
$M$ effective weights (Fig.~\ref{fig:arch}). In contrast, the reported
fully joint results depend on covariance eigenvalue truncation. At the
$10^{-9}$ relative cutoff used in Section~\ref{sec:per-platform-numerics},
the retained mode count grows from $8$ at $N=1$ to $48$ at $N=729$,
so these are regularization-dependent values, not a geometry-only
ceiling for the full $MN$-element array.

Let $\sigma_p$ be the residual per-platform range error,
with phase error $\sigma_\phi=k\sigma_p$. The nominal
simulation gives $0.15$~dB mean-SIR loss at $\sigma_p=4.7$~mm,
so the $0.15$~dB budget calls for roughly $5$~mm residual
range precision \emph{under the adopted independent Gaussian
position-error model}, not a complete synchronization specification. This is comparable to the sub-$5$~mm
formation-baseline accuracy reported in~\cite{ref:allende2018baseline}.
Over a $1$~ms interval, the allowable residual carrier
offset is $32$~Hz (fractional $1.6\times10^{-8}$ at $2$~GHz),
while the reported oscillator Allan deviation measures a
different stability property~\cite{ref:uso2013}. Inter-platform
synchronization remains essential~\cite{ref:formation}.
Geometry-based compensation removes the predictable
$25$~kHz differential Doppler and up to $61~\mu$s of delay.
Finally, target offsets of $3$ and $10$~cm along the baseline
reduce mean SIR by $0.18$ and $2.1$~dB, respectively.

% \subsection{Performance Metrics}\label{sec:metrics}
% We use two signal-to-interference ratios, each
% aggregating the leakage profile $|\ba(\bq)^H\bw|^2$ over $\Gset$
% differently:
% \begin{align}
% \text{SIR}_{\text{mean}}(\bw) &\triangleq \frac{|\ba(\bq^\star)^H\bw|^2}{|\Gset|^{-1}\!\int_\Gset |\ba(\bq)^H\bw|^2 \, d\bq}, \label{eq:sir-mean}\\
% \text{SIR}_{\text{peak}}(\bw) &\triangleq \frac{|\ba(\bq^\star)^H\bw|^2}{\sup_{\bq \in \Gset}|\ba(\bq)^H\bw|^2}. \label{eq:sir-peak}
% \end{align}
% \rev{Here $d\bq$ is a reference measure on $\Gset$, length or area in Section~\ref{sec:numerics}, and $|\Gset|$ its total.} The mean captures average leakage (spectrum coexistence), the peak
% worst-case leakage (protected receivers). We use $\text{SIR}_{\text{mean}}$ as
% the design objective because its quadratic denominator admits a
% Rayleigh-quotient analysis (Theorem~\ref{thm:bound}), and report
% $\text{SIR}_{\text{peak}}$ as a diagnostic.

\subsection{Performance Metrics}\label{sec:metrics}
We consider two signal-to-interference ratio (SIR) metrics that
aggregate the leakage profile $|\ba(\bq)^H\bw|^2$ over $\Gset$
in different ways:
\begin{align}
\text{SIR}_{\text{mean}}(\bw) &\triangleq \frac{|\ba(\bq^\star)^H\bw|^2}{|\Gset|^{-1}\!\int_\Gset |\ba(\bq)^H\bw|^2 \, d\bq}, \label{eq:sir-mean}\\
\text{SIR}_{\text{peak}}(\bw) &\triangleq \frac{|\ba(\bq^\star)^H\bw|^2}{\sup_{\bq \in \Gset}|\ba(\bq)^H\bw|^2}. \label{eq:sir-peak}
\end{align}
Here $d\bq$ denotes the reference measure on $\Gset$, taken as
length or area in Section~\ref{sec:numerics}, and $|\Gset|$ denotes
its total measure. The mean SIR measures average leakage relevant
to spectrum coexistence, while the peak SIR captures worst-case
leakage at protected receivers. The mean SIR provides a geometry-dependent benchmark through
the generalized Rayleigh quotient of Theorem~\ref{thm:bound}.
The SOCP formulations instead maximize target gain with
peak-leakage control; their optimal objective values need not
coincide with the maximum mean SIR. We report
$\text{SIR}_{\text{peak}}$ as a complementary diagnostic.

% \subsection{Reference SOCP Formulation}
% \rev{The reference design balances coherent target gain against worst-case leakage over $\Gset$, subject to a power budget} $P_{\text{tot}}$:
% \begin{subequations}\label{eq:socp}
% \begin{align}
% \max_{\bw,\, t_g \ge 0}\quad &\re\{\ba_t^H \bw\} - \beta\, t_g
%    \label{eq:socp-obj}\\
% \text{s.t.}\quad &\im\{\ba_t^H\bw\} = 0, \label{eq:socp-real}\\
% & |\ba(\bq)^H\bw| \le t_g, \quad \forall \bq \in \Gset,
%    \label{eq:socp-leak}\\
% & \|\bw\|_2^2 \le P_{\text{tot}}, \label{eq:socp-pow}
% \end{align}
% \end{subequations}
% where $\beta \ge 0$ trades target gain against suppression. The auxiliary variable $t_g$ \rev{caps} the worst-case field magnitude on
% $\Gset$, so~\eqref{eq:socp-leak} is equivalent to
% $t_g \ge \sup_{\bq \in \Gset}|\ba(\bq)^H\bw|$. The equality $\im\{\ba_t^H\bw\} = 0$ removes the global $U(1)$ phase ambiguity intrinsic to the SIR objective. Under it\rev{, and since $\bw=\mathbf{0}$ is feasible, every optimum has $|\ba_t^H\bw| = \re\{\ba_t^H\bw\}\ge0$}, so the SOCP template applies without sacrificing optimality.

\subsection{Reference SOCP Formulation}
The reference design balances coherent target gain against
worst-case leakage over $\Gset$ under a total power budget
$P_{\text{tot}}$:
\begin{subequations}\label{eq:socp}
\begin{align}
\max_{\bw,\, t_g \ge 0}\quad &\re\{\ba_t^H \bw\} - \beta\, t_g
   \label{eq:socp-obj}\\
\text{s.t.}\quad &\im\{\ba_t^H\bw\} = 0, \label{eq:socp-real}\\
& |\ba(\bq)^H\bw| \le t_g, \quad \forall \bq \in \Gset,
   \label{eq:socp-leak}\\
& \|\bw\|_2^2 \le P_{\text{tot}}, \label{eq:socp-pow}
\end{align}
\end{subequations}
where $\beta \ge 0$ controls the trade-off between target gain and
interference suppression. The auxiliary variable $t_g$ bounds
the worst-case field magnitude over $\Gset$, so~\eqref{eq:socp-leak}
is equivalent to
$t_g \ge \sup_{\bq \in \Gset}|\ba(\bq)^H\bw|$.
The constraint $\im\{\ba_t^H\bw\} = 0$ removes the global
$U(1)$ phase ambiguity of the SIR objective.
Under this constraint, and since $\bw=\mathbf{0}$ is feasible,
every optimum satisfies
$|\ba_t^H\bw| = \re\{\ba_t^H\bw\}\ge0$.
Thus, the SOCP formulation introduces no loss of optimality.

Problem~\eqref{eq:socp} imposes one constraint for every $\bq \in \Gset$
and is therefore semi-infinite. Solving it on a finite sample set
$\{\bq_k\}_{k=1}^K \subset \Gset$ does not guarantee feasibility
between samples, where the field may exceed the prescribed threshold
even if the sampled problem is feasible. The cutting-plane scheme
of Section~\ref{sec:socp} adds the worst-case violators to the working
set until the residual leakage falls below a specified tolerance.
Problem~\eqref{eq:socp} serves as the reference formulation for both
numerical implementation using standard convex solvers and the
subsequent analysis. To examine its underlying structure, we use
the equivalent \emph{canonical form} in Section~\ref{sec:dual}:
\begin{subequations}\label{eq:canonical}
\begin{align}
\max_{\bw}\quad &\re\{\ba_t^H\bw\} \label{eq:canonical-obj}\\
\text{s.t.}\quad &|\ba(\bq)^H\bw|^2 \le \tau, \quad \forall\,\bq\in\Gset,
   \label{eq:canonical-leak}\\
& \|\bw\|_2^2 \le P_{\text{tot}}, \label{eq:canonical-pow}
\end{align}
\end{subequations}
Here the soft penalty $\beta\,t_g$ in~\eqref{eq:socp} is replaced
by a hard threshold $\tau$ on the squared field magnitude.
The two formulations represent the same target-gain and leakage
trade-off rather than being pointwise identical.
Problems~\eqref{eq:socp} and~\eqref{eq:canonical} are, respectively,
the weighted-sum and $\varepsilon$-constraint scalarizations of a
common bi-objective problem and can recover corresponding supported
trade-off solutions under convexity and Slater qualification. For every $\tau > 0$ and $P_{\rm tot}>0$, strong duality
implies that a non-negative penalty weight can be chosen so that
the two problems admit a common optimizer $\bw^\star$.
Specifically, if $\mu^\star$ is an optimal multiplier measure
for the squared-leakage constraints of~\eqref{eq:canonical},
one such weight is $\beta=2\sqrt{\tau}\,\mu^\star(\Gset)$, since
$(\bw^\star,\sqrt{\tau})$ then satisfies the KKT conditions
of~\eqref{eq:socp} with leakage multiplier $2\sqrt{\tau}\mu^\star$.
This guarantees a common optimizer, not identical optimizer sets.
We use the hard form for the dual analysis because its pointwise
constraints admit a single non-negative measure $\mu$ on $\Gset$
as the dual variable. The soft form retains the auxiliary scalar
$t_g$, which is convenient for numerical solvers but less suitable
for exposing the dual structure.

\section{Lagrangian Dual Analysis}\label{sec:dual}

This section establishes the optimal beamformer structure,
the finite-support dual certificate, the mean-SIR ceiling,
and robustness to perturbations. Three Hermitian matrices recur below. The \emph{leakage covariance}
and \emph{moment matrix} are
\begin{equation}\label{eq:Rg-def}
\begin{aligned}
\Rg &\triangleq \frac{1}{|\Gset|}\int_\Gset {\ba(\bq)\ba(\bq)^H} d\bq,\\
\bM(\mu) &\triangleq \int_\Gset {\ba(\bq)\ba(\bq)^H} d\mu(\bq).
\end{aligned}
\end{equation}
The geometry fixes $\Rg$
and hence the mean-SIR ceiling of Section~\ref{sec:bound},
whereas $\bM(\mu)$ depends on the dual measure and reduces
to $|\Gset|\Rg$ when $d\mu(\bq)=d\bq$ (the unnormalized uniform measure). The effective covariance
$\Reff$ below combines $\bM(\mu)$ with the power-budget
multiplier and determines the optimal beamformer.

For brevity, define
\begin{equation}\label{eq:Reff-def}
\Reff(\mu, \nu) \triangleq \bM(\mu) + \nu\bI
= \int_\Gset \ba(\bq)\ba(\bq)^H\, d\mu(\bq) + \nu\bI,
\end{equation}
which combines the dual variables into a Hermitian operator.

The pointwise constraint $|\ba(\bq)^H\bw|^2 \le \tau$
in~\eqref{eq:canonical} is a normalized power-flux-density-type
surrogate, not a regulatory limit in $\mathrm{W/m^2}$, whose
conversion would require element patterns, impedance and absolute
field calibration. As in SOCP focusing models~\cite{ref:vazquez2016},
it is imposed pointwise over the protected region rather than on
average, which makes the problem semi-infinite and introduces the
dual measure $\mu$.

In Problem~\eqref{eq:canonical}, the non-negative measure $\mu$
multiplies the continuum of leakage constraints and the scalar
$\nu\ge0$ multiplies the power budget. By complementary slackness,
$\mathrm{supp}(\mu^\star)$ contains only points where the leakage
constraint is active, and Theorem~\ref{thm:sparsity} shows that
$\mu^\star$ can be chosen with at most $2M-1$ atoms. These are
\emph{active worst-case protected-region points} determined by
optimality rather than prescribed in advance. In array-processing
terms, $\bM(\mu)$ acts as the interference covariance of the active
leakage points, and $\nu\bI$ in~\eqref{eq:Reff-def} as endogenous
diagonal loading, with $\nu^\star$, the shadow price of transmit
power, set by the power constraint rather than tuned.

% \subsection{Dual Formulation and Constraint Qualification}\label{sec:prelim}
% \rev{We fix the space the dual variable lives in and verify the constraint
% qualification that strong duality requires.}

\subsection{Dual Formulation and Constraint Qualification}\label{sec:prelim}

% The continuum of leakage constraints in~\eqref{eq:canonical} forces $\mu$
% into \rev{an infinite-dimensional space}. We model it as a finite
% non-negative Borel measure on $\Gset$, write $\mathcal{M}_+(\Gset)$ for the
% cone of such measures and equip it with the weak-$*$ topology. That is the right notion here because the limit is a discrete Dirac mass while its approximants are diffuse, which is the situation in Theorem~\ref{thm:sparsity}. Two standard
% facts are used without further comment: the closed unit ball in
% $\mathcal{M}_+(\Gset)$ is weak-$*$ compact (Banach--Alaoglu) and \rev{the dual
% function $g(\mu,\nu)$, obtained from the Lagrangian~\eqref{eq:Lagr-quad} of
% Appendix~\ref{app:thm-structure},}
% is weak-$*$ lower semi-continuous \rev{and, since $g\ge\tau\mu(\Gset)+\nu P_{\rm tot}$,
% has bounded sublevel sets}, so the dual optimum $d^\star$ is attained.
% Operationally, $\mu$ is a non-negative weight on ground points\rev{; an
% optimal $\mu$ exists with} at most \rev{$2M-1$} point masses.

The continuum of leakage constraints in~\eqref{eq:canonical} requires
$\mu$ to belong to an infinite-dimensional space. We represent
it as a finite non-negative Borel measure on $\Gset$, denote the cone
of such measures by $\mathcal{M}_+(\Gset)$, and equip it with the
weak-$*$ topology. This topology allows diffuse measures to converge
to discrete Dirac masses, as considered in
Theorem~\ref{thm:sparsity}. We use two standard results: the closed
unit ball in $\mathcal{M}_+(\Gset)$ is weak-$*$ compact
(Banach--Alaoglu), and the dual function $g(\mu,\nu)$,
derived from the Lagrangian~\eqref{eq:Lagr-quad} in
Appendix~\ref{app:thm-structure}, is weak-$*$ lower
semi-continuous and, since
$g\ge\tau\mu(\Gset)+\nu P_{\rm tot}$, has bounded sublevel sets.
The dual optimum $d^\star$ is therefore attained.
Operationally, $\mu$ assigns non-negative weights to protected-region
locations, and an optimal $\mu$ exists with at most
$2M-1$ point masses.

% Problem~\eqref{eq:canonical} satisfies the \emph{Slater-Robinson constraint
% qualification}~\cite{ref:shapiro2009} if there exists $\bw_0 \in \Cset^M$ with
% \begin{equation}\label{eq:slater}
% |\ba(\bq)^H\bw_0|^2 < \tau\ \forall \bq \in \Gset, \quad \|\bw_0\|^2 < P_{\text{tot}}.
% \end{equation}
% \rev{A non-degenerate certificate exists whenever $\tau,P_{\rm tot}>0$: take
% $\bw_0 = \delta\,\ba_t/\|\ba_t\|$ with
% $0 < \delta < \min\{\sqrt{P_{\rm tot}},\,
% \sqrt{\tau}/\sup_{\bq\in\Gset}\|\ba(\bq)\|\,\}$, finite because
% $\ba(\cdot)$ is continuous on the compact set $\Gset$. Then
% $|\ba(\bq)^H\bw_0| \le \delta\,\|\ba(\bq)\| < \sqrt{\tau}$ on $\Gset$ and
% $\|\bw_0\|^2 = \delta^2 < P_{\rm tot}$, so both constraints are strictly slack
% while transmit power and target gain remain non-zero.}

Problem~\eqref{eq:canonical} satisfies the \emph{Slater--Robinson constraint
qualification}~\cite{ref:shapiro2009} if there exists $\bw_0 \in \Cset^M$
such that
\begin{equation}\label{eq:slater}
|\ba(\bq)^H\bw_0|^2 < \tau\ \forall \bq \in \Gset, \quad \|\bw_0\|^2 < P_{\text{tot}}.
\end{equation}
A nonzero strictly feasible point exists for any $\tau,P_{\rm tot}>0$.
Specifically, choose $\bw_0 = \delta\,\ba_t/\|\ba_t\|$ with
$0 < \delta < \min\{\sqrt{P_{\rm tot}},\,
\sqrt{\tau}/\sup_{\bq\in\Gset}\|\ba(\bq)\|\,\}$.
The supremum is finite because $\ba(\cdot)$ is continuous on the compact
set $\Gset$. It follows that
$|\ba(\bq)^H\bw_0| \le \delta\,\|\ba(\bq)\| < \sqrt{\tau}$ on $\Gset$ and
$\|\bw_0\|^2 = \delta^2 < P_{\rm tot}$. Thus, both constraints are
strictly satisfied while the transmit power and target gain remain
nonzero.

% \subsection{Optimal Structure and Support Sparsity}\label{sec:structure}
% \rev{The first result gives the form of every optimal beamformer; the second
% shows that an optimal dual measure can be chosen finitely atomic.}

\subsection{Optimal Structure and Support Sparsity}\label{sec:structure}
The first result characterizes every optimal beamformer, while the
second establishes the existence of an optimal dual measure with finite
support.

% \begin{theorem}[Optimal Linear Beamformer Structure]\label{thm:structure}
% Suppose $\Gset$ is compact, $\ba : \Gset \to \Cset^M$ is continuous, and~\eqref{eq:slater}
% holds. Then~\eqref{eq:canonical} attains its optimum, strong duality holds,
% dual optima $(\mu^\star,\nu^\star)\in\mathcal{M}_+(\Gset)\times\Rset_+$ exist,
% \rev{and every primal optimum $\bw^\star$ satisfies the stationarity condition}
% \rev{\begin{equation}\label{eq:wstar-kkt}
% \Reff(\mu^\star,\nu^\star)\,\bw^\star \;=\; \tfrac{1}{2}\ba_t .
% \end{equation}}
% \rev{If in addition $\Reff(\mu^\star,\nu^\star)\succ0$, the primal optimum is
% unique and admits the closed form}
% \begin{equation}\label{eq:wstar}
% \bw^\star = \tfrac{1}{2}\Reff(\mu^\star, \nu^\star)^{-1}\, \ba_t.
% \end{equation}
% \end{theorem}

\begin{theorem}[Optimal Linear Beamformer Structure]\label{thm:structure}
Suppose $\Gset$ is compact, $\ba : \Gset \to \Cset^M$ is continuous,
and~\eqref{eq:slater} holds. Then Problem~\eqref{eq:canonical}
attains its optimum, strong duality holds, and dual optimal solutions
$(\mu^\star,\nu^\star)\in\mathcal{M}_+(\Gset)\times\Rset_+$ exist.
Moreover, every primal optimum $\bw^\star$ satisfies the
stationarity condition
\begin{equation}\label{eq:wstar-kkt}
\Reff(\mu^\star,\nu^\star)\,\bw^\star \;=\; \tfrac{1}{2}\ba_t .
\end{equation}
If, in addition, $\Reff(\mu^\star,\nu^\star)\succ0$, the primal
optimum is unique and has the closed-form expression
\begin{equation}\label{eq:wstar}
\bw^\star = \tfrac{1}{2}\Reff(\mu^\star, \nu^\star)^{-1}\, \ba_t.
\end{equation}
\end{theorem}

\begin{IEEEproof}
See Appendix~\ref{app:thm-structure}.
\end{IEEEproof}

% \begin{remark}[When the closed form applies]\label{rem:reff-pd}
% \rev{Positive definiteness of $\Reff(\mu^\star,\nu^\star)$ is not implied
% by~\eqref{eq:slater}. It holds whenever $\nu^\star>0$, which by complementary
% slackness requires the power constraint to be active, and whenever the
% steering vectors carried by $\mathrm{supp}(\mu^\star)$ span $\Cset^M$; failing
% both, \eqref{eq:wstar-kkt} fixes $\bw^\star$ only modulo $\ker(\Reff^\star)$.
% At the operating point of Section~\ref{sec:numerics} the power constraint is
% active to within $2.4\times10^{-10}$ of the budget, and \eqref{eq:wstar}
% evaluated with the computed dual reproduces $\bw^\star$ to a
% relative error of $5.4\times10^{-5}$.}
% \end{remark}

\begin{remark}[When the closed form applies]\label{rem:reff-pd}
Positive definiteness of $\Reff(\mu^\star,\nu^\star)$ does not follow
from~\eqref{eq:slater}. It holds if $\nu^\star>0$, which by complementary
slackness requires the power constraint to be active. It also holds if
the steering vectors in $\mathrm{supp}(\mu^\star)$ span $\Cset^M$.
If neither condition holds, \eqref{eq:wstar-kkt} determines $\bw^\star$
only up to a component in $\ker(\Reff^\star)$.
At the operating point of Section~\ref{sec:numerics}, the power
constraint is active to within $2.4\times10^{-10}$ of the budget,
and \eqref{eq:wstar}, evaluated using the computed dual variables,
reproduces $\bw^\star$ with a relative error of $5.4\times10^{-5}$.
\end{remark}

%\subsection{Support Sparsity}

% \begin{theorem}[Sparsity of Optimal Dual Measure]\label{thm:sparsity}
% Under the hypotheses of Theorem~\ref{thm:structure}, there exists an
% optimal dual measure $\mu^\star$ for~\eqref{eq:canonical} that is
% purely atomic with finite support:
% \begin{equation}\label{eq:atomic}
% \mu^\star = \sum_{i=1}^{K} c_i\, \delta_{\bq_i^\star},
% \quad c_i > 0,\ \bq_i^\star \in \Gset^\star,\ \rev{K \le 2M-1},
% \end{equation}
% where $K$ denotes the cardinality of the atomic support of $\mu^\star$,
% $\delta_{\bq}$ is the Dirac measure at $\bq$, and
% $\Gset^\star \triangleq \{\bq \in \Gset : |\ba(\bq)^H\bw^\star|^2 = \tau\}$
% is the active set \rev{of a primal optimum $\bw^\star$}. \rev{The bound exceeds by one the
% real dimension, $2M-2$, of the affine set $\{\bx\in\Cset^M:\bw^{\star H}\bx=\tau\}$, which
% contains $\ba(\bq)\ba(\bq)^H\bw^\star$ for every $\bq\in\Gset^\star$.}
% \end{theorem}

% Representation~\eqref{eq:atomic} is \emph{a posteriori}, not an \emph{a
% priori} sampling rule: the active points are an output of the optimization,
% coinciding with the worst-case leakage directions, and a pre-selected
% discretization of $\Gset$ generally does not contain them.

% \begin{IEEEproof}
% See Appendix~\ref{app:thm-structure}.
% \end{IEEEproof}

\begin{theorem}[Sparsity of Optimal Dual Measure]\label{thm:sparsity}
Under the hypotheses of Theorem~\ref{thm:structure}, an optimal dual
measure $\mu^\star$ for~\eqref{eq:canonical} can be chosen to be
purely atomic with finite support:
\begin{equation}\label{eq:atomic}
\mu^\star = \sum_{i=1}^{K} c_i\, \delta_{\bq_i^\star},
\quad c_i > 0,\ \bq_i^\star \in \Gset^\star,\ K \le 2M-1,
\end{equation}
where $K$ is the number of atoms, $\delta_{\bq}$ denotes the Dirac
measure at $\bq$, and
$\Gset^\star \triangleq \{\bq \in \Gset : |\ba(\bq)^H\bw^\star|^2 = \tau\}$
is the active set associated with a primal optimum $\bw^\star$.
The bound is one greater than the real dimension, $2M-2$, of
the affine set $\{\bx\in\Cset^M:\bw^{\star H}\bx=\tau\}$, which
contains $\ba(\bq)\ba(\bq)^H\bw^\star$ for every
$\bq\in\Gset^\star$.
\end{theorem}

Representation~\eqref{eq:atomic} is \emph{a posteriori}, not an
\emph{a priori} sampling rule. The active points are determined by
the optimization and correspond to the worst-case leakage directions;
a pre-selected discretization of $\Gset$ generally does not contain
these points.

\begin{IEEEproof}
See Appendix~\ref{app:thm-structure}.
\end{IEEEproof}

\begin{remark}[Scope and tightness of the finite-support certificate]
\label{rem:sparsity-use}\label{rem:sparsity-tight}
Theorem~\ref{thm:sparsity} establishes a finite-support dual measure,
but it neither provides an algorithm nor bounds the working set of the
cutting-plane method in Section~\ref{sec:socp}, which retains all
appended constraints, and a sparse certificate does not imply a small
active set. At twice the Nyquist density of
Section~\ref{sec:grid-resolution}, the computed design has leakage
peaks at only $x=\pm11.12$~km, while the cutting-plane method retains
$441$ constraints, exceeding $2M-1=15$. Thus, a finite atomic
representation can coexist with a much larger working set, and
Theorem~\ref{thm:sparsity} does not prescribe a sampling strategy.

The mass $\mu(\Gset)$, weighted by $\tau$ in the dual objective,
requires no separate count because every $\mu$ supported on
$\Gset^\star$ satisfies
$\bw^{\star H}\bM(\mu)\bw^\star=\tau\mu(\Gset)$.
Theorem~\ref{thm:sparsity} preserves optimality rather than the
entire matrix $\bM(\mu^\star)$; reproducing that matrix requires
at most $M^2$ atoms by the same argument. For uniform linear arrays
in the plane-wave limit, the Carath\'eodory--Toeplitz
theorem~\cite{ref:carath-toeplitz} represents $\bM(\mu)$ using
$M$ atoms on the unit circle, which need not correspond to points
in $\Gset$. In the tested designs, the leakage threshold is reached
at few points. Even for min-max designs, the count never exceeds
$14$ across the $50$ geometries of Remark~\ref{rem:two-params},
below $2M-1=15$.
\end{remark}

\subsection{Performance Limits and Geometric Degeneracy}\label{sec:bound}\label{sec:degeneracy}
The following bound applies to any array geometry, while the
near-collinear analysis that follows establishes its scaling behavior.

\begin{theorem}[Mean-SIR Upper Bound]\label{thm:bound}
Suppose $\Gset$ is compact, every non-empty relatively open
$\mathcal{U}\subseteq\Gset$ has $|\mathcal{U}|>0$, $\ba(\cdot)$ is continuous,
and $\{\ba(\bq) : \bq \in \Gset\}$ spans $\Cset^M$.
Then $\Rg \succ 0$, and for every
$\bw \in \Cset^M\setminus\{\mathbf{0}\}$,
\begin{equation}\label{eq:bound-strict}
\frac{|\ba_t^H\bw|^2}{\bw^H \Rg \bw} \le \rho^\star \triangleq \ba_t^H \Rg^{-1} \ba_t,
\end{equation}
with equality if and only if $\bw \propto \Rg^{-1}\ba_t$.
\end{theorem}

\begin{IEEEproof}
For any $\bx \in \Cset^M$,
$\bx^H \Rg \bx = |\Gset|^{-1}\!\int_\Gset |\ba(\bq)^H \bx|^2 d\bq \ge 0$,
so $\Rg \succeq 0$. If $\bx^H\Rg\bx=0$, the non-negative continuous
integrand vanishes throughout $\Gset$. The spanning hypothesis then
implies $\bx = \mathbf{0}$, establishing $\Rg \succ 0$.
Substituting $\bu = \Rg^{1/2}\bw$, the Cauchy--Schwarz inequality gives
$|\ba_t^H\bw|^2 =
|\langle \Rg^{-1/2}\ba_t, \bu \rangle|^2 \le \|\Rg^{-1/2}\ba_t\|^2 \|\bu\|^2
= \rho^\star\, \bw^H\Rg\bw$, with equality if and only if
$\bu \propto \Rg^{-1/2}\ba_t$, or equivalently
$\bw \propto \Rg^{-1}\ba_t$.
\end{IEEEproof}

% \begin{corollary}[Spectral Decomposition and Harmonic Mean Form]\label{cor:spectral}
% Let $\Rg = \sum_{i=1}^M \lambda_i \bv_i \bv_i^H$. Then
% \begin{equation}\label{eq:spectral}
% \rho^\star = \sum_{i=1}^{M} \frac{|\bv_i^H \ba_t|^2}{\lambda_i} = \frac{\|\ba_t\|^2}{H_\alpha},
% \end{equation}
% where $H_\alpha \triangleq (\sum_i \alpha_i \lambda_i^{-1})^{-1}$ is the
% $\alpha$-weighted harmonic mean of the eigenvalues with
% $\alpha_i = |\bv_i^H \ba_t|^2 / \|\ba_t\|^2$. Hence $\rho^\star$ is
% maximized when $\{\alpha_i\}$ concentrate on the smallest eigenvalues.
% \end{corollary}

% \begin{IEEEproof}
% \rev{Substitute $\Rg^{-1}=\sum_i\lambda_i^{-1}\bv_i\bv_i^H$
% into~\eqref{eq:bound-strict} and use $\sum_i\alpha_i=1$.}
% \end{IEEEproof}

\begin{corollary}[Spectral Decomposition and Harmonic Mean Form]\label{cor:spectral}
Let $\Rg = \sum_{i=1}^M \lambda_i \bv_i \bv_i^H$. Then
\begin{equation}\label{eq:spectral}
\rho^\star = \sum_{i=1}^{M} \frac{|\bv_i^H \ba_t|^2}{\lambda_i} = \frac{\|\ba_t\|^2}{H_\alpha},
\end{equation}
where $H_\alpha \triangleq (\sum_i \alpha_i \lambda_i^{-1})^{-1}$
is the $\alpha$-weighted harmonic mean of the eigenvalues, with
$\alpha_i = |\bv_i^H \ba_t|^2 / \|\ba_t\|^2$. Hence,
$\rho^\star$ is maximized when the weights $\{\alpha_i\}$
are concentrated on the smallest eigenvalues.
\end{corollary}

\begin{IEEEproof}
Substituting $\Rg^{-1}=\sum_i\lambda_i^{-1}\bv_i\bv_i^H$
into~\eqref{eq:bound-strict} and using $\sum_i\alpha_i=1$
gives the result.
\end{IEEEproof}

%\subsection{Geometric Degeneracy: Quantitative Asymptotic Analysis}\label{sec:degeneracy}

% Theorem~\ref{thm:bound} does not reveal how $\rho^\star$ depends on $D$, $L$,
% $h_t$, $h_s$ or scales with $M$. This subsection makes the dependence explicit in
% the near-collinear regime of Section~\ref{sec:near-collinear}, for the nominal
% geometry: the positions $\{\bp_m\}$ are deterministic and exactly known, and
% $\varepsilon \triangleq \max\{D/h_s, L/h_s, h_t/h_s\}$ measures the intrinsic
% angular degeneracy between target and protection region. Robustness to deviations
% from the nominal positions is treated in Section~\ref{sec:robustness}.

Theorem~\ref{thm:bound} establishes the mean-SIR ceiling $\rho^\star$ but
does not show how it depends on $D$, $L$, $h_t$, $h_s$, or scales with $M$.
We examine these dependencies in the near-collinear regime of
Section~\ref{sec:near-collinear}, assuming nominal geometry with
deterministic and exactly known platform positions $\{\bp_m\}$.
The parameter $\varepsilon \triangleq \max\{D/h_s, L/h_s, h_t/h_s\}$
characterizes the intrinsic angular degeneracy between the target
and the protection region. Robustness to deviations from the nominal
positions is considered separately in Section~\ref{sec:robustness}.

\begin{theorem}[Quantitative Asymptotic Bound]\label{thm:asymp}
Consider $M\ge2$ element positions $\bp_m = (x_m,0,h_s)$
with $|x_m|\le D/2$, distinct $x_m$, and
$\max_m x_m-\min_m x_m=D$,
$\Gset=\{(x,0,0):|x|\le L\}$, $L>0$,
target $\bq^\star=(0,0,h_t)$, $0\le h_t<h_s$, and
$\varepsilon \triangleq \max\{D/h_s,\,L/h_s,\,h_t/h_s\}$.
Define the Fresnel parameter $\Phi$, the minimum element
separation $d_{\min}$, and its associated Fresnel parameter
$\Phi_{\min}$ as
\begin{equation}\label{eq:fresnel-param}
\begin{aligned}
\Phi &\triangleq \frac{k\,D\,L}{h_s},
&\quad d_{\min} &\triangleq \min_{m\neq n}|x_m-x_n| > 0,\\
\Phi_{\min} &\triangleq \frac{k\,L\,d_{\min}}{h_s}. & &
\end{aligned}
\end{equation}
Then, parts (a) and (b) hold in a joint limit with $\varepsilon\to0$ and
$M/\Phi_{\min}\to0$, where the electrical-separation parameter
$\Phi_{\min}$ must grow faster than $M$. One valid path has
fixed $M,h_s$; $D,L,d_{\min}=\Theta(\varepsilon h_s)$,
$h_t=\mathcal O(\varepsilon h_s)$ and
$\lambda=\Theta(\varepsilon^3h_s)$, so
$\Phi_{\min}=\Theta(\varepsilon^{-1})$:
\begin{enumerate}
\item[(a)] \begin{equation}\label{eq:Rg-expansion}
\Rg=h_s^{-2}(\bI+\bDelta),\qquad
\bDelta\triangleq h_s^2\Rg-\bI
\end{equation}
where $\bDelta$ is Hermitian and
$\|\bDelta\|_{\rm op}=\mathcal{O}(\varepsilon^2+M/\Phi_{\min})$.
The off-diagonal entries satisfy
\begin{equation}\label{eq:offdiag}
\max_{m\neq n}\frac{\bigl|[\Rg]_{mn}\bigr|}{[\Rg]_{mm}}
\;=\;\mathcal{O}\!\left(\Phi_{\min}^{-1}\right),
\end{equation}
which becomes $\mathcal{O}((M-1)/\Phi)$ for a uniformly
spaced 1D arc with $d_{\min}=D/(M-1)$.
\item[(b)] \begin{equation}\label{eq:asymp}
\rho^\star \;=\; M\left(\frac{h_s}{h_s-h_t}\right)^{\!2}
\bigl(1+\mathcal{O}(\varepsilon^2+M/\Phi_{\min})\bigr),
\end{equation}
or, in decibels, $\rho^\star_{\rm dB}=10\log_{10}M
+20\log_{10}\frac{h_s}{h_s-h_t}+\delta$, with $|\delta|$
given in Table~\ref{tab:asymp-error}.
\item[(c)] For arbitrary parameter values, let
$\bar s=D/(2h_s)$, $t=h_t/h_s$, $\ell=L/h_s$,
$p=\bar s^2t(2-t)/[(1-t)^2+\bar s^2]$, $q=(1+\bar s^2)\ell^2/12$ and
$\bar\beta=CS_M/\Phi_{\min}$, where
$S_M=\max_{0\le j<M}(H_j+H_{M-1-j})\le2(1+\ln M)$, the $H_j$ being the
harmonic numbers with $H_0=0$, and
$C=(1+w^2)^{5/2}+\ell w[2(1+w^2)^2+3(1+w^2)^{7/2}]$ with $w=\ell+\bar s$.
If $\bar\beta<1$ and $4q<1$, then
\[
\frac{1-p}{(1+q)(1+\bar\beta)}\;\le\;
\frac{\rho^\star}{M\bigl(\frac{h_s}{h_s-h_t}\bigr)^{2}}\;\le\;
\frac{1}{(1-4q)(1-\bar\beta)}.
\]
\end{enumerate}
\end{theorem}

\begin{IEEEproof}
See Appendix~\ref{app:prop-asymp}.
\end{IEEEproof}

% \begin{table}[!t]
% \revfloat
% \centering\small
% \caption{\rev{Residual $|\delta|$ in Theorem~\ref{thm:asymp}(b),
% maximized over $M\in\{2,4,\dots,128\}$, at $f=2$~GHz, $h_s=600$~km,
% $h_t=10$~km, $L=20$~km, on grids certified converged to
% $5\times10^{-3}$~dB. The last column is the bound guaranteed by
% Theorem~\ref{thm:asymp}(c), maximized over the same $M$.}}
% \label{tab:asymp-error}
% \setlength{\tabcolsep}{4pt}
% \begin{tabular}{@{}rcccc@{}}
% \toprule
% $D$ (km) & $\varepsilon=D/h_s$ & $\Phi=kDL/h_s$ & $\max_M|\delta|$ (dB) & \rev{Bound (dB)} \\
% \midrule
% $50$   & $0.083$ & $7.0\times10^{4}$ & $0.0024$ & \rev{$0.0791$} \\
% $200$  & $0.333$ & $2.8\times10^{5}$ & $0.0026$ & \rev{$0.0257$} \\
% $300$  & $0.500$ & $4.2\times10^{5}$ & $0.0076$ & \rev{$0.0250$} \\
% $700$  & $1.167$ & $9.8\times10^{5}$ & $0.0376$ & \rev{$0.0515$} \\
% $1000$ & $1.667$ & $1.4\times10^{6}$ & $0.0610$ & \rev{$0.0792$} \\
% \bottomrule
% \end{tabular}
% \end{table}

\begin{table}[!t]
\centering\small
\caption{Residual of Theorem~\ref{thm:asymp}(b) and nonasymptotic upper bound of part (c), maximized over $M\in\{2,4,\ldots,128\}$. Grid tolerance: $5\times10^{-3}$~dB.}
\label{tab:asymp-error}
\setlength{\tabcolsep}{4pt}
\begin{tabular}{@{}rcccc@{}}
\toprule
$D$ (km) & $\varepsilon=D/h_s$ & $\Phi=kDL/h_s$ & $\max_M|\delta|$ (dB) & Bound (dB) \\
\midrule
$50$   & $0.083$ & $7.0\times10^{4}$ & $0.0024$ & $0.0791$ \\
$200$  & $0.333$ & $2.8\times10^{5}$ & $0.0026$ & $0.0257$ \\
$300$  & $0.500$ & $4.2\times10^{5}$ & $0.0076$ & $0.0250$ \\
$700$  & $1.167$ & $9.8\times10^{5}$ & $0.0376$ & $0.0515$ \\
$1000$ & $1.667$ & $1.4\times10^{6}$ & $0.0610$ & $0.0792$ \\
\bottomrule
\end{tabular}
\end{table}

\begin{remark}[Which parameter controls the approximation]
\label{rem:two-params}
The two parameters have distinct roles, and neither is sufficient on
its own. Over a $256$-fold range of $\Phi$
($f\in\{0.5,2,8\}$~GHz, $D\in\{75,300\}$~km,
$L\in\{5,20\}$~km), the measured
$\max_{m\neq n}|[\Rg]_{mn}|/[\Rg]_{mm}$
follows $\Phi^{-0.89}$ with correlation $-0.98$, consistent
with~\eqref{eq:offdiag}. Meanwhile, $\kappa(\Rg)$ is $1.004$ at
$\varepsilon=0.125$ and $1.061$ at $\varepsilon=0.5$,
independently of $\Phi$. At the operating point,
$\varepsilon=0.5$ is not small, yet Table~\ref{tab:asymp-error}
shows a residual of only $0.0076$~dB, as explained by part~(c) of
Theorem~\ref{thm:asymp}. Beyond the first-order term
in~\eqref{eq:asymp}, the baseline enters mainly through
$p\approx(h_t/h_s)(D/h_s)^2/2$, which stays small at $D/h_s=0.5$
because each platform views the target and ground from nearly the
same angle. At $M=8$, $p=2.0\times10^{-3}$,
$q=9.8\times10^{-5}$ and $\bar\beta=8.3\times10^{-5}$
give the guaranteed range $-0.0095\le\delta\le0.0021$~dB.

The approximation also remains accurate beyond the layouts assumed in
its derivation, where $\bp_m=(x_m,0,h_s)$ places the platform--target
and platform--interferer directions in a common near-nadir cone.
We evaluate $50$ geometries, including a curved circular-orbit
alternative to the straight reference with the same chord, Gaussian
cross-track and radial perturbations up to
$10^{-2}h_s$, and random three-dimensional clouds with departures
from collinearity reaching $0.23\,h_s$. Across these geometries,
$\rho^\star$ remains within $0.0053$~dB of
$10\log_{10}M + 20\log_{10}[h_s/(h_s-h_t)]$, while
$\kappa(\Rg)$ increases only from $1.06$ to $1.13$. The approximation
thus holds for layouts sharing the near-collinear \emph{viewing}
geometry $h_t \ll h_s$, even when the \emph{array} is not collinear.
\end{remark}

In the near-collinear regime of
Section~\ref{sec:near-collinear}, where $D/h_s$, $L/h_s$ and $h_t/h_s$
are all small and $M/\Phi_{\min}\ll1$, the achievable SIR is
therefore governed primarily by the array order $M$ rather than
the aperture span $D$.

\begin{remark}[Per-platform spatial processing is bounded by the same geometric ceiling]
\label{rem:per-platform}
Replacing $[\ba(\bq)]_m$ with
$\sqrt{G_m(\bq)}[\ba(\bq)]_m$ in
Theorem~\ref{thm:bound} gives
\begin{equation}\label{eq:rho-gen}
\rho^\star_{\rm gen,dB}
\approx 10\log_{10}M
+20\log_{10}\frac{h_s}{h_s-h_t}
+10\log_{10}\bar G_{\rm eff},
\end{equation}
where $\bar G_{\rm eff}$ represents the platform-level discrimination
between $\bq^\star$ and $\Gset$, defined as the linear ratio
of the per-platform gain toward $\bq^\star$ to its gain over
$\Gset$. The separation is exact when the gain is identical across
platforms and constant over $\Gset$, and approximate otherwise.
The factor $\bar G_{\rm eff}$ in~\eqref{eq:rho-gen}
concerns transmit-side local-array discrimination only.
The receive-side comparison in Fig.~\ref{fig:arch}(c) is instead
an illustrative link-level desired-to-leakage ratio,
$\rho_{\rm link,dB}^{\rm eff}=\rho^\star_{\rm dB}+G_{R,{\rm dB}}$,
conditional on a linear desired-receiver gain $G_{R,{\rm lin}}$
and normalized unit receive gain at protected receivers, where
$G_{R,{\rm dB}}=10\log_{10}G_{R,{\rm lin}}$.
It does not change the transmit-field bound in
Theorem~\ref{thm:bound}. In contrast, a transmit-side local aperture
views both within the same near-collinear cone and provides
at most $1.2$~dB with any tested well-posed local strategy.
Thus, increasing the per-platform aperture alone cannot raise
the tested transmit-field ceiling substantially. Transmit-side
cases remain covered by Theorem~\ref{thm:bound} applied to
the corresponding effective array,
with the hierarchical architecture following the
$10\log_{10}M$ scaling of Theorem~\ref{thm:asymp}.
The distinction lies in the number of coherently combinable
elements $\Mtot$ relative to the effective rank
$r_{\rm eff} \triangleq (\tr\bR)^2/\tr(\bR^2) \le \Mtot$
of the ground covariance $\bR$. When $r_{\rm eff} \approx \Mtot$,
as in the hierarchical design ($r_{\rm eff}=8.0$ at $\Mtot=M=8$),
each element contributes a full-strength mode, allowing additional
interference suppression. When $r_{\rm eff} \ll \Mtot$, as under
joint processing ($r_{\rm eff}\le 9.2$ for $\Mtot\le 2048$),
additional elements contribute mainly weak, ill-conditioned modes.
Fully joint processing across all $MN$ elements
can raise $\rho^\star$ by several decibels above $10\log_{10}M$
by increasing the design dimension, but requires network-wide
phase synchronization.
\end{remark}

\subsection{Robustness to Position Perturbation}\label{sec:robustness}
Finite orbit-determination accuracy causes the actual positions
$\hat\bp_m = \bp_m + \boldsymbol{\xi}_m$ to differ from their nominal
design values.
\begin{remark}[Position-perturbation sensitivity]\label{rem:perturbation}
Let $\boldsymbol{\xi}_m \in \Rset^3$ be independent and identically
distributed (i.i.d.) isotropic zero-mean Gaussian random vectors with
$\E[\boldsymbol{\xi}_m\boldsymbol{\xi}_m^H] = \sigma_p^2\bI_3$.
The projection of $\boldsymbol{\xi}_m$ onto any line of sight then has
standard deviation $\sigma_p$, corresponding to the residual range error
of Section~\ref{sec:architecture}. The induced phase error has standard
deviation $\sigma_\phi \triangleq k\sigma_p$.
For beamforming weights fixed at their nominal values, the Cox
phase-perturbation identity~\cite{ref:cox1987} (Appendix~\ref{app:perturb}) gives
the exact ratio of mean target power to mean leakage power under
the phase-only perturbation approximation:
\begin{equation}\label{eq:perturb-loss}
\bar\rho\triangleq\frac{\E|\hat\ba_t^H\bw|^2}{\E[\bw^H\hat{\bR}_{\!g}\bw]}
=\frac{e^{-\sigma_\phi^2}S+(1-e^{-\sigma_\phi^2})S_d}{e^{-\sigma_\phi^2}I+(1-e^{-\sigma_\phi^2})I_d},
\end{equation}
where $S=|\ba_t^H\bw|^2$, $I=\bw^H\Rg\bw$, and $S_d$, $I_d$
are obtained by replacing $\ba_t\ba_t^H$ and $\Rg$ with their
respective diagonal matrices. Thus, each platform's own power
contribution remains unaffected by the phase errors. For the MVDR
weights at the operating point, $I_d/I-1<10^{-4}$ and $S_d/S=1/M$,
giving $\bar\rho/\rho^\star\approx1/M+(1-1/M)e^{-\sigma_\phi^2}$.
The corresponding loss increases as
$(10/\ln10)(1-1/M)\sigma_\phi^2$ and approaches
$10\log_{10}M=9.0$~dB, the incoherent-combining level.

% 替换 Remark[Position-Perturbation Sensitivity] 中关于式 (17) 适用范围的表述
% 回应：Reviewer 1 关于鲁棒性、以及本文实测（Fig. 7）
% \rev{The Monte Carlo simulation of Fig.~\ref{fig:robustness}(b), over
% $\sigma_p/\lambda \in [10^{-3},1]$ with weights fixed at the nominal geometry, averages the SIR in decibels instead. The two agree within $0.05$~dB for $\sigma_\phi \le 0.35$~rad. At $\sigma_p/\lambda=0.1$ ($\sigma_\phi=0.63$~rad), \eqref{eq:perturb-loss} gives $1.46$ against a simulated $1.64$~dB. For large $\sigma_\phi$ the decibel average saturates at $11.3$~dB rather than $9.0$~dB, the mean of a logarithm lying below the logarithm of the mean. The Gaussian-phase law $(10/\ln10)\sigma_\phi^2$, which omits the $1/M$ term, exceeds both and diverges beyond $\sigma_\phi=1$~rad, i.e.\ $\sigma_p = \lambda/2\pi = 24$~mm at
% $2$~GHz, comfortably above differential global navigation satellite system (GNSS) accuracy.}
% \end{remark}

% \noindent \rev{Appendix~\ref{app:perturb} derives~\eqref{eq:perturb-loss} from Cox's identity.}

The Monte Carlo simulation in Fig.~\ref{fig:robustness}(b) considers
$\sigma_p/\lambda \in [10^{-3},1]$ with weights fixed at the nominal
geometry and averages the SIR in decibels rather than taking the ratio
of expected powers. The two measures agree within $0.05$~dB for
$\sigma_\phi \le 0.35$~rad. At $\sigma_p/\lambda=0.1$
($\sigma_\phi=0.63$~rad), \eqref{eq:perturb-loss} predicts a loss
of $1.46$~dB, compared with $1.64$~dB from simulation. For large
$\sigma_\phi$, the decibel-averaged loss approaches $11.3$~dB
rather than $9.0$~dB because the mean of the logarithm is smaller
than the logarithm of the mean. The Gaussian-phase law
$(10/\ln10)\sigma_\phi^2$, which neglects the $1/M$ term,
overestimates both losses and continues to increase beyond
$\sigma_\phi=1$~rad, corresponding to
$\sigma_p = \lambda/2\pi = 24$~mm at $2$~GHz, well above
differential global navigation satellite system (GNSS) accuracy.
\end{remark}

% \textbf{Denominator (output interference).}

\section{Beamforming Realizations and Structural Connections}\label{sec:algos}

% \subsection{Roadmap: From Dual Structure to Algorithms}\label{sec:algo-roadmap}

% \subsection{Roadmap: From Dual Structure to Algorithms}
% \label{sec:algo-roadmap}

% \rev{Theorem~\ref{thm:structure} exhibits a measure-parameterized family of
% effective covariances $\Reff(\mu,\nu)$, each producing a beamformer
% $\tfrac12\Reff(\mu,\nu)^{-1}\ba_t$. The optimal $\mu^\star$ is the member
% solving~\eqref{eq:canonical}; other, prescribed members reproduce beamformers
% already in use. Uniform $\mu$ on $\Gset$ gives MVDR; that uniform part plus a
% finite sum of atoms gives the discretely constrained designs of which LCMV is the standard
% instance; and the cutting-plane scheme of Section~\ref{sec:socp} is the one
% method here whose objective converges to the optimum of~\eqref{eq:socp}
% (Theorem~\ref{thm:cp-conv}). The constant-modulus
% algorithm is not a member of the family, optimizing over a non-convex
% manifold on which Theorem~\ref{thm:structure} does not apply, and enters only
% as a practical realization whose gap to the bound we measure. The connections
% are structural and imply no shared optimum. Per-solve costs are $\mathcal{O}(M^3)$ for MVDR, $\mathcal{O}(M^3+K^3)$ for
% LCMV with $K$ nulls, $\mathcal{O}(K_kM^2+M^3)$ per outer iteration for the
% cutting-plane scheme, and $\mathcal{O}(M^3+TM^2)$ for $T$ Riemannian iterations;
% measured latencies are in Table~\ref{tab:numerics}. Theorem~\ref{thm:bound}
% applies to all four designs, Theorem~\ref{thm:structure} to the three convex
% ones.}

\subsection{Roadmap: From Dual Structure to Algorithms}
\label{sec:algo-roadmap}

Theorem~\ref{thm:structure} characterizes the effective
covariance $\Reff(\mu^\star,\nu^\star)$ at dual optimality.
When an effective covariance is positive definite, the
corresponding Lagrangian maximizer has direction
$\Reff(\mu,\nu)^{-1}\ba_t$ (and scale $1/2$ in~\eqref{eq:wstar}).
Other prescribed measures give structural analogues, not
necessarily primal optima of~\eqref{eq:canonical}.
The normalized uniform measure $d\mu=d\bq/|\Gset|$
gives the MVDR direction, while a
uniform component combined with finitely many atoms leads to
discretely constrained designs, with LCMV as a standard example.
The cutting-plane scheme of Section~\ref{sec:socp} is the only
method considered here whose objective is proved to converge to
the optimum of~\eqref{eq:socp} (Theorem~\ref{thm:cp-conv}).
The constant-modulus algorithm lies outside this family because
it optimizes over a non-convex manifold to which
Theorem~\ref{thm:structure} does not apply, whereas
Theorem~\ref{thm:bound} applies to all four designs.
The computational costs are
$\mathcal{O}(M^3)$ for MVDR, $\mathcal{O}(M^3+K^3)$ for LCMV
with $K$ nulls, and $\mathcal{O}(M^3+TM^2)$ for
$T$ Riemannian iterations, excluding covariance assembly.
Cutting-plane SOCP has solver-dependent complexity, with
$\mathcal{O}(K_kM^2+M^3)$ operations per dense interior-point
Newton step. Measured latencies are reported in
Table~\ref{tab:numerics}.

% \subsection{MVDR: Unconstrained Mean-SIR Optimum}\label{sec:mvdr}

% Rather than imposing the pointwise hard constraints of~\eqref{eq:canonical},
% the classical MVDR formulation minimizes the average ground power
% $\bw^H\Rg\bw$ subject only to a target-gain constraint, yielding the
% closed-form solution
% \begin{equation}\label{eq:mvdr}
% \bw_{\text{MVDR}} = \Rg^{-1}\ba_t.
% \end{equation}
% This matches~\eqref{eq:wstar} under $\mu = |\Gset|^{-1}d\bq$, $\nu = 0$, so
% the dual-parameterized family contains MVDR as a limiting member. MVDR is not \rev{in general} an optimum of~\eqref{eq:canonical}\rev{: it controls the mean leakage rather than the peak, and in Fig.~\ref{fig:protection} its peak envelope rises above the SOCP envelope's maximum at six ground cells}. In Algorithm~\ref{alg:mvdr} the population
% covariance $\Rg$ of~\eqref{eq:Rg-def} is replaced by the sample covariance
% \rev{$\bA_g^H \bA_g / N_{\rm g}$ plus a small diagonal loading
% $\epsilon\,\tr(\Rg)/M \cdot \bI$ for conditioning, where the $n$th row of $\bA_g$ is $\ba(\bq_n)^H$ for $N_{\rm g}$ ground
% points $\bq_n$ drawn uniformly over $\Gset$}. We retain MVDR as the
% unconstrained-SIR benchmark; complexity \rev{$\mathcal{O}(N_{\rm g}M^2 + M^3)$}.

\subsection{MVDR: Unconstrained Mean-SIR Optimum}\label{sec:mvdr}

Instead of imposing the pointwise constraints of~\eqref{eq:canonical},
classical MVDR minimizes the average ground power $\bw^H\Rg\bw$
subject to a target-gain constraint, giving the closed-form solution
\begin{equation}\label{eq:mvdr}
\bw_{\text{MVDR}} = \Rg^{-1}\ba_t.
\end{equation}
This corresponds to~\eqref{eq:wstar} with
$d\mu = |\Gset|^{-1}d\bq$ and $\nu = 0$, so MVDR belongs to
the dual-parameterized family as a limiting case. However, MVDR is
not generally optimal for~\eqref{eq:canonical}, since it minimizes
mean rather than peak leakage. In Fig.~\ref{fig:protection},
its peak envelope exceeds the maximum of the SOCP envelope at six
ground cells. In Algorithm~\ref{alg:mvdr}, the population covariance
$\Rg$ of~\eqref{eq:Rg-def} is approximated by the sample covariance
$\bA_g^H \bA_g / N_{\rm g}$ with a small diagonal loading
$\epsilon\,\tr(\Rg)/M \cdot \bI$ for numerical conditioning,
where the $n$th row of $\bA_g$ is $\ba(\bq_n)^H$ for
$N_{\rm g}$ ground points $\bq_n$ sampled uniformly over $\Gset$.
We use MVDR as the unconstrained mean-SIR benchmark, with complexity
$\mathcal{O}(N_{\rm g}M^2 + M^3)$.

% \begin{algorithm}[!h]
% \small
% \caption{MVDR Beamformer}\label{alg:mvdr}
% \begin{algorithmic}[1]
% \Require $\ba_t$, ground samples \rev{$\bA_g \in \Cset^{N_{\rm g} \times M}$}, $P_{\text{tot}}$
% \State \rev{$\Rg \leftarrow \bA_g^H \bA_g / N_{\rm g}$} \Comment{sample covariance}
% \State $\Rg \leftarrow \Rg + \epsilon\, \tr(\Rg)/M \cdot \bI$ \Comment{diagonal loading}
% \State $\bw \leftarrow \Rg^{-1}\, \ba_t$
% \State $\bw \leftarrow \bw \sqrt{P_{\text{tot}} / (\bw^H\bw)}$
% \State \Return $\bw$
% \end{algorithmic}
% \end{algorithm}

\begin{algorithm}[!h]
\small
\caption{MVDR Beamformer}\label{alg:mvdr}
\begin{algorithmic}[1]
\Require $\ba_t$, ground samples $\bA_g \in \Cset^{N_{\rm g} \times M}$, $P_{\text{tot}}$
\State $\Rg \leftarrow \bA_g^H \bA_g / N_{\rm g}$ \Comment{sample covariance}
\State $\Rg \leftarrow \Rg + \epsilon\, \tr(\Rg)/M \cdot \bI$ \Comment{diagonal loading}
\State $\bw \leftarrow \Rg^{-1}\, \ba_t$
\State $\bw \leftarrow \bw \sqrt{P_{\text{tot}} / (\bw^H\bw)}$
\State \Return $\bw$
\end{algorithmic}
\end{algorithm}

\subsection{LCMV: Discrete Atomic Measure}\label{sec:lcmv}

Consider $d\mu(\bq) = |\Gset|^{-1}d\bq + \sum_{i=1}^K \gamma_i \delta_{\bq_i}$
with $\nu = 0$, giving
$\Reff = \Rg + \sum_i \gamma_i \ba(\bq_i)\ba(\bq_i)^H$.
As $\gamma_i\to\infty$, the inverse-covariance direction
in~\eqref{eq:wstar} approaches the null-constrained
minimum-variance direction by the Woodbury identity,
provided that $\ba(\bq_i)$ are linearly independent and
$\ba_t$ does not lie in their span. After normalization to
$\ba_t^H\bw=1$, the limiting direction coincides with the
LCMV/Frost beamformer~\cite{ref:frost1972}:
\begin{equation}\label{eq:lcmv}
\bw_{\text{LCMV}} = \Rg^{-1}\bC \big(\bC^H\Rg^{-1}\bC\big)^{-1}\bef,
\end{equation}
where $\bC = [\ba_t, \ba(\bq_1),\dots,\ba(\bq_K)]$ and
$\bef = [1,0,\dots,0]^T$.
The uniform component of $\mu$ provides the background
covariance $\Rg$; with a purely atomic measure, this role
would instead be played by $\nu\bI$.
Algorithm~\ref{alg:lcmv} summarizes the procedure.
\begin{algorithm}[ht]
\small
\caption{LCMV Beamformer}\label{alg:lcmv}
\begin{algorithmic}[1]
\Require $\ba_t$, ground samples $\bA_g \in \Cset^{N_{\rm g} \times M}$, constraint atoms $\{\bq_i\}_{i=1}^K \subset \Gset$, $P_{\text{tot}}$
\State $\Rg \leftarrow \bA_g^H \bA_g / N_{\rm g}$ \Comment{sample covariance}
\State $\Rg \leftarrow \Rg + \epsilon\, \tr(\Rg)/M \cdot \bI$ \Comment{diagonal loading}
\State $\bC \leftarrow [\ba_t,\, \ba(\bq_1),\, \dots,\, \ba(\bq_K)]$
\State $\bef \leftarrow [1, 0, \dots, 0]^T$
\State $\bw \leftarrow \Rg^{-1}\bC(\bC^H\Rg^{-1}\bC)^{-1}\bef$
\State $\bw \leftarrow \bw\sqrt{P_{\text{tot}}/(\bw^H\bw)}$
\State \Return $\bw$
\end{algorithmic}
\end{algorithm}
\begin{remark}[LCMV failure under geometric degeneracy]\label{rem:lcmv-fail}
When $\bq_i$ are chosen at ground leakage peaks, their steering
vectors $\ba(\bq_i)$ are highly correlated with $\ba_t$.
Imposing the hard null constraints $\ba(\bq_i)^H\bw = 0$
then also reduces the target gain $\ba_t^H\bw$, potentially
causing severe signal attenuation. For the equally spaced
nulls considered in Section~\ref{sec:numerics}, the condition
number $\kappa(\bC^H\Rg^{-1}\bC)$ increases from
$\sim 22$ at $K=4$ to beyond $10^{16}$
at $K=7$, indicating near-singularity.
\end{remark}

% \subsection{SOCP with Cutting Plane}\label{sec:socp}

% \subsection{Cutting-Plane SOCP and Convergence}\label{sec:socp}

% The general SOCP~\eqref{eq:socp} \rev{is solved by}
% constraint generation: an outer loop augments a finite working set
% $\Gset_k \subset \Gset$ with the \rev{$n_{\rm add}$ worst} violators located by an inner
% oracle and re-solves until no significant violation remains
% (Algorithm~\ref{alg:socp}).

% \begin{algorithm}[t]
% \small
% \caption{Cutting-Plane SOCP}\label{alg:socp}
% \begin{algorithmic}[1]
% \Require initial coarse grid $\Gset_0 \subset \Gset$, evaluation grid $\Gset_{\text{eval}}$ ($N_{\rm eval}$ points, Remark~\ref{rem:grids}), \rev{batch size $n_{\rm add}$}
% \State $k \leftarrow 0$
% \Repeat
% \State Solve~\eqref{eq:socp} on $\Gset_k$ for $(\bw_k, t_{g,k})$
% \State $\bq^* \leftarrow \argmax_{\bq \in \Gset_{\text{eval}}} |\ba(\bq)^H\bw_k|$
% \If{$|\ba(\bq^*)^H\bw_k| \le t_{g,k} + \epsilon$} \textbf{break}\EndIf
% \State \rev{$\Gset_{k+1} \leftarrow \Gset_k \cup \mathcal{V}_k$, with $\mathcal{V}_k \ni \bq^*$ the $n_{\rm add}$ largest violators on $\Gset_{\text{eval}}$}, $k \leftarrow k+1$
% \Until{converged}
% \State \Return $\bw_k$
% \end{algorithmic}
% \end{algorithm}

% The next result establishes asymptotic convergence and finite termination
% for any prescribed tolerance. \rev{It rests on the cutting-plane argument
% below, not on Theorem~\ref{thm:sparsity}, which bounds the atomic
% representation of $\mu^\star$ and not the working set: in
% Section~\ref{sec:numerics} the scheme retains $|\Gset_k| = 441$ constraints,
% against at most $2M-1=15$ atoms.}

\subsection{Cutting-Plane SOCP and Convergence}\label{sec:socp}

The general SOCP~\eqref{eq:socp} is solved using
constraint generation. An outer loop adds the $n_{\rm add}$ largest
constraint violators identified by an inner oracle to a finite working
set $\Gset_k \subset \Gset$, then re-solves the problem until no
significant violation remains (Algorithm~\ref{alg:socp}).

\begin{algorithm}[t]
\small
\caption{Cutting-Plane SOCP}\label{alg:socp}
\begin{algorithmic}[1]
\Require initial coarse grid $\Gset_0\subset\Gset$, evaluation grid
$\Gset_{\rm eval}$ (Remark~\ref{rem:grids}), batch size
$n_{\rm add}$, relative tolerance $\epsilon_{\rm rel}=10^{-3}$
\State $k \leftarrow 0$
\Repeat
\State Solve~\eqref{eq:socp} on $\Gset_k$ for $(\bw_k, t_{g,k})$
\State $\bq^* \leftarrow \argmax_{\bq \in \Gset_{\text{eval}}} |\ba(\bq)^H\bw_k|$
\If{$|\ba(\bq^*)^H\bw_k| \le (1+\epsilon_{\rm rel})t_{g,k}$}
\textbf{break}\EndIf
\State $\Gset_{k+1}\leftarrow\Gset_k\cup\mathcal V_k$,
with $\mathcal V_k$ the $n_{\rm add}$ largest \emph{new} violating
points on $\Gset_{\rm eval}$, including $\bq^*$; discard duplicates,
$k\leftarrow k+1$
\Until{converged}
\State \Return $\bw_k$
\end{algorithmic}
\end{algorithm}

The following theorem establishes asymptotic convergence and finite
termination for any prescribed tolerance. Its proof does not
use Theorem~\ref{thm:sparsity}, which bounds the atoms of $\mu^\star$
rather than the working set (Remark~\ref{rem:sparsity-use}).

\begin{theorem}[Asymptotic Convergence and Finite $\varepsilon$-Termination]\label{thm:cp-conv}
Let $\Gset$ be compact, $\ba(\cdot)$ continuous, and $\beta>0$.
Let $J^\star$ denote the optimal value of the semi-infinite
program~\eqref{eq:socp}, and $J_k$ its optimal value over the working
set $\Gset_k$ at outer iteration $k$ of Algorithm~\ref{alg:socp} continued without its stopping test.
Denote the corresponding solution by $(\bw_k,t_{g,k})$ and the
worst violator by
$\bq_k=\arg\max_{\bq\in\Gset}|\ba(\bq)^H\bw_k|$.
For the finite oracle used in Algorithm~\ref{alg:socp},
replace $\Gset$ by $\Gset_0\cup\Gset_{\rm eval}$ and interpret
$J^\star$ as the optimum over this finite set. Closeness to the
continuous problem depends on the grid resolution discussed in
Section~\ref{sec:grid-resolution}. Then:
\begin{enumerate}
\item[(a)] (Monotonicity)
$\{J_k\}$ is monotone non-increasing: $J_{k+1} \le J_k$.
\item[(b)] (Asymptotic convergence)
$J_k \downarrow J^\star$ as $k \to \infty$.
\item[(c)] (Finite $\varepsilon$-termination)
For any $\varepsilon > 0$, there exists a finite
$K_\varepsilon<\infty$ such that
$|\ba(\bq_k)^H\bw_k| - t_{g,k} \le \varepsilon$
for all $k \ge K_\varepsilon$.
\end{enumerate}
\end{theorem}

\begin{IEEEproof}
See Appendix~\ref{app:prop-cpconv}.
\end{IEEEproof}

\begin{remark}[Convergence rate]\label{rem:cp-rate}
Standard Kelley-type cutting-plane methods for continuous semi-infinite
programs do not generally have a worst-case linear convergence rate.
Sublinear convergence with occasional zigzag behavior has been
documented~\cite{ref:kelley1960,ref:hettich1993}.
We make no convergence-rate claim for Algorithm~\ref{alg:socp},
which uses constraint exchange rather than Newton-type updates.
Its relative test with $\epsilon_{\rm rel}=10^{-3}$, rather than the
absolute $\varepsilon$ of Theorem~\ref{thm:cp-conv}, is met after
three outer iterations for the segment of Table~\ref{tab:numerics}
and five for the region of Fig.~\ref{fig:protection}, with working
sets of $441$ and $1701$ constraints, including $401$ and $1681$
initial grid points, respectively.
\end{remark}

% \subsection{Manifold Optimization for Constant Modulus}\label{sec:manifold}

% \begin{revblock}
% The constant-modulus problem differs from~\eqref{eq:canonical} in both objective
% and feasible set, so the optimal structure of Theorem~\ref{thm:structure} does not
% directly apply. We evaluate Algorithm~\ref{alg:manifold} against the upper bound
% $\rho^\star$; the algorithm itself provides no global optimality guarantee.
% \end{revblock}

% \rev{On the complex torus $\Mset_{\text{CM}} = (\sqrt{P_m}\,\Tset)^M$ of
% Section~\ref{sec:cm}~\cite{ref:absil2008}, Algorithm~\ref{alg:manifold}
% maximizes the mean SIR $f(\bw) =
% |\ba_t^H\bw|^2 / \bw^H\Rg\bw$ of~\eqref{eq:sir-mean}, whose Euclidean gradient \rev{under the inner product $\re\{\bu^H\bv\}$, twice the Wirtinger derivative, is}}
% \begin{equation}
% \rev{\nabla f = 2\Big(\frac{\ba_t \ba_t^H \bw}{\bw^H\Rg\bw} - \frac{|\ba_t^H\bw|^2}{(\bw^H\Rg\bw)^2}\Rg\bw\Big).}
% \end{equation}
% Projecting onto $T_\bw \Mset_{\text{CM}}$:
% \begin{equation}\label{eq:riemgrad}
% \text{grad}\, f = \rev{\nabla f} - \re\!\big(\rev{\nabla f} \odot \bw^* \oslash |\bw|^2\big) \odot \bw,
% \end{equation}
% with $\odot,\oslash$ element-wise. Element-wise normalization
% $\Retr_\bw(\eta\,\bg) = \sqrt{P_m}(\bw + \eta\bg) \oslash |\bw + \eta\bg|$
% retracts onto $\Mset_{\text{CM}}$.

\subsection{Manifold Optimization for Constant Modulus}\label{sec:manifold}

The constant-modulus problem differs from~\eqref{eq:canonical}
in both its objective and feasible set, so the optimal structure
of Theorem~\ref{thm:structure} does not directly apply.
We compare Algorithm~\ref{alg:manifold} with the upper bound
$\rho^\star$, although the algorithm itself offers no guarantee
of global optimality.

On the complex torus
$\Mset_{\text{CM}} = (\sqrt{P_m}\,\Tset)^M$
defined in Section~\ref{sec:cm}~\cite{ref:absil2008},
Algorithm~\ref{alg:manifold} maximizes the mean SIR
$f(\bw) = |\ba_t^H\bw|^2 / \bw^H\Rg\bw$
of~\eqref{eq:sir-mean}. Its Euclidean gradient
under the inner product $\re\{\bu^H\bv\}$, equal to twice
the Wirtinger derivative, is
\begin{equation}
\nabla f = 2\Big(\frac{\ba_t \ba_t^H \bw}{\bw^H\Rg\bw} - \frac{|\ba_t^H\bw|^2}{(\bw^H\Rg\bw)^2}\Rg\bw\Big).
\end{equation}
Projecting this gradient onto $T_\bw \Mset_{\text{CM}}$ gives
\begin{equation}\label{eq:riemgrad}
\text{grad}\, f = \nabla f - \re\!\big(\nabla f \odot \bw^* \oslash |\bw|^2\big) \odot \bw,
\end{equation}
where $\odot$ and $\oslash$ denote element-wise multiplication
and division, respectively. The element-wise normalization
$\Retr_\bw(\eta\,\bg) = \sqrt{P_m}(\bw + \eta\bg) \oslash |\bw + \eta\bg|$
defines a retraction onto $\Mset_{\text{CM}}$.

% \begin{revblock}
% Successive directions are combined
% by the Polak--Ribi\`ere$^+$ rule, with the tangent-space projection
% $\Pi_{\bw}(\bz)=\bz-\re(\bz\odot\bw^*\oslash|\bw|^2)\odot\bw$
% of~\eqref{eq:riemgrad} acting as vector transport and a restart at steepest
% ascent whenever the combined direction is not an ascent direction. The step
% size is obtained by Armijo backtracking along the retraction curve
% ($c_1 = 10^{-4}$, $\varrho = 1/2$, at most $40$ halvings), each iteration
% starting from twice the previously accepted step. Iteration stops when the
% scale-free first-order measure $\|\text{grad}f\|\,\|\bw\|/f$ or the relative
% change in $f$ falls below $\epsilon_{\rm tol} = 10^{-12}$, or after $T = 500$ iterations. The iteration counts in Remark~\ref{rem:mvdr-init} and Table~\ref{tab:numerics} follow this stopping rule; every run stops on the relative-change test except two random-phase draws of Remark~\ref{rem:mvdr-init}, which reach $T$.
% \end{revblock}

Successive search directions are combined by the
Polak--Ribi\`ere$^+$ rule, with the tangent-space projection
$\Pi_{\bw}(\bz)=\bz-\re(\bz\odot\bw^*\oslash|\bw|^2)\odot\bw$
of~\eqref{eq:riemgrad} as vector transport, and the method restarts
from steepest ascent whenever ascent fails. Armijo backtracking along
the retraction curve sets the step, starting from twice the
previously accepted step. If Armijo fails after $40$ halvings, the
step is rejected and the current iterate is returned. The algorithm
stops when $\|\text{grad}f\|\,\|\bw\|/f$ or the relative change in
$f$ falls below $\epsilon_{\rm tol}$, or after $T$ iterations, with
the parameters of Algorithm~\ref{alg:manifold}. All runs stop by the
relative-change test, except two random-phase starts in
Remark~\ref{rem:mvdr-init} that reach $T$.

\begin{remark}[MVDR-projected initialization]\label{rem:mvdr-init}
Initialization mainly affects convergence speed rather than
the attained SIR. Across $200$ random-phase starts, the SIR
remains within $1.1\times10^{-7}$~dB of the optimum, while
iterations range from $8$ to $T$. The maximum-ratio
transmission (MRT) start $\bw_0\propto\ba_t/|\ba_t|$
is optimal when $\Rg$ is diagonal. For the near-scalar
$\Rg$ considered here ($\kappa(\Rg)=1.06$), MRT lies
$3\times10^{-8}$~dB below the optimum. We instead initialize
by projecting MVDR onto $\Mset_{\rm CM}$:
\begin{equation}\label{eq:mvdr-init}
\bw_0 = \sqrt{P_m}\,\bw_{\text{MVDR}}
\oslash|\bw_{\text{MVDR}}|.
\end{equation}
The MVDR amplitudes vary by at most $1.6\%$ on the
converged covariance. This initialization converges in $1$
iteration, compared with $9$ for MRT and a median of $25$
over the $200$ random-phase starts.
\end{remark}

The complete procedure, using MVDR initialization followed by
Riemannian conjugate-gradient updates, is given in
Algorithm~\ref{alg:manifold}.

\begin{algorithm}[t]
\small
\caption{Riemannian CG with MVDR Initialization and Armijo Backtracking}\label{alg:manifold}
\begin{algorithmic}[1]
\Require $\ba_t, \Rg, P_m$; $c_1 = 10^{-4}$, $\varrho = 1/2$ ($\le 40$ backtracks), $\epsilon_{\rm tol} = 10^{-12}$, $T = 500$
\State $\bw_{\text{MVDR}} \leftarrow \Rg^{-1}\ba_t$ \Comment{closed-form, $\mathcal{O}(M^3)$}
\State $\bw_0 \leftarrow \sqrt{P_m}\, \bw_{\text{MVDR}} \oslash |\bw_{\text{MVDR}}|$ \Comment{$\oslash$: element-wise division}
\State $\bg_0 \leftarrow \text{grad}f(\bw_0)$ via~\eqref{eq:riemgrad}; $\bd_0 \leftarrow \bg_0$; $\eta \leftarrow 0$
\For{$k = 0, \dots, T-1$}
\If{$\re(\bd_k^H\bg_k) \le 0$} $\bd_k \leftarrow \bg_k$ \Comment{restart}\EndIf
\If{$\|\bg_k\|\,\|\bw_k\|/f(\bw_k) \le \epsilon_{\rm tol}$} \textbf{break}\EndIf
\If{$\eta=0$} $\eta \leftarrow \|\bw_k\|/\|\bd_k\|$ \Comment{initial step}
\EndIf
\State $b\leftarrow0$
\While{$b<40$ and
$f(\Retr_{\bw_k}(\eta\bd_k)) < f(\bw_k)+c_1\eta\,\re(\bd_k^H\bg_k)$}
\State $\eta\leftarrow\varrho\eta$; $b\leftarrow b+1$
\EndWhile
\If{$f(\Retr_{\bw_k}(\eta\bd_k)) < f(\bw_k)+c_1\eta\,\re(\bd_k^H\bg_k)$}
\textbf{break} \Comment{reject step on Armijo failure}
\EndIf
\State $\bw_{k+1} \leftarrow \Retr_{\bw_k}(\eta\,\bd_k)$; $\eta \leftarrow 2\eta$ \Comment{accepted step}
\State $\bg_{k+1} \leftarrow \text{grad}f(\bw_{k+1})$ via~\eqref{eq:riemgrad}
\State $\beta_k \leftarrow \max\{0,\ \re\big(\bg_{k+1}^H(\bg_{k+1}-\Pi_{\bw_{k+1}}(\bg_k))\big)/\|\bg_k\|^2\}$
\State $\bd_{k+1} \leftarrow \bg_{k+1} + \beta_k\,\Pi_{\bw_{k+1}}(\bd_k)$ \Comment{transport}
\If{$|f(\bw_{k+1})-f(\bw_k)|/f(\bw_k) < \epsilon_{\rm tol}$} \textbf{break}\EndIf
\EndFor
\State \Return $\hat\bw\,e^{-j\angle(\ba_t^H\hat\bw)}$, $\hat\bw$ the last iterate \Comment{phase-aligned}
\end{algorithmic}
\end{algorithm}

%% ============================================================================
% \section{Numerical Results}\label{sec:numerics}

% \subsection{Experimental Setup}

% We evaluate the limiting case $h_g \to 0$\rev{, in which $\Gset_{\rm gnd}$ lies in
% the ground plane}. The formation is $M=8$ platforms uniformly
% spaced on a $D=300$~km arc at $h_s=600$~km, each carrying a $27\times27$ planar
% array ($N_m=729$ half-wavelength elements, $\approx1.95$~m square at $f=2$~GHz)
% with fixed boresight analog beamforming toward $\bq^\star$. Its beamwidth
% ($\approx4.4^\circ$) exceeds the target--protection separation
% ($\approx1.9^\circ$), so the array gain is nearly uniform over
% $\{\bq^\star\}\cup\Gset$ (Remark~\ref{rem:per-platform}), the local aperture is a
% separable factor and $\bw\in\Cset^M$ is the only optimization variable. The target is at $\bq^\star=(0,0,10)$~km with
% $P_{\rm tot}=64$~kW\rev{. The protection region is the segment
% $\Gset=\{(x,0,0):|x|\le L\}$, $L=20$~km, except in Figs.~\ref{fig:protection}(a)
% and~\ref{fig:arch}(b), which use the square $[-L,L]^2\times\{0\}$}.

\section{Numerical Results}\label{sec:numerics}

\begin{figure*}[!t]
\centering
\includegraphics[width=0.95\textwidth]{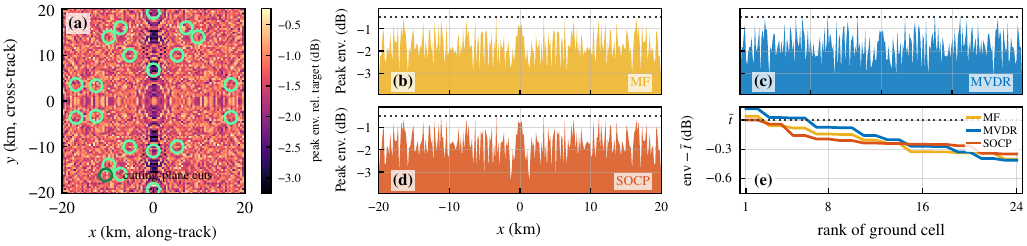}
\vspace{-10pt}
\caption{Ground-region peak leakage envelopes. (a) Two-dimensional cellwise maximum with cutting-plane locations. (b)--(d) One-dimensional envelopes for matched-filter, MVDR, and SOCP designs; dotted line: SOCP reference $\bar t$. (e) Ranked envelope margins.}
\label{fig:protection}
\end{figure*}

\subsection{Experimental Setup}

We consider the limiting case $h_g \to 0$, where
$\Gset_{\rm gnd}$ lies in the ground plane. The formation consists
of $M=8$ platforms uniformly spaced along the straight,
common-altitude $D=300$~km reference baseline at $h_s=600$~km
(the ``1D arc'' of Figs.~\ref{fig:system} and~\ref{fig:arch}).
A circular-orbit geometry is examined separately in Remark~\ref{rem:two-params}. Table~\ref{tab:numerics} and
Figs.~\ref{fig:protection}--\ref{fig:robustness} use the
$N_m=\nobreak1$ platform-center response. For the aperture comparison
in Section~\ref{sec:per-platform-numerics}, each platform
carries a $27\times27$ planar array ($N_m=729$ half-wavelength
elements, $\approx1.95$~m square at $f=2$~GHz) with fixed
boresight analog beamforming toward $\bq^\star$.
The local beamwidth ($\approx4.4^\circ$) exceeds the
$\approx1.9^\circ$ ground-region half-angular extent as viewed from orbit.
Fig.~\ref{fig:per-platform} uses complex spherical-wave
responses for the local arrays (with validated numerical
compression for joint processing). The target
is located at $\bq^\star=(0,0,10)$~km, with
$P_{\rm tot}=64$~kW. The protection region is the segment
$\Gset=\{(x,0,0):|x|\le L\}$, $L=20$~km, except in
Figs.~\ref{fig:protection}(a) and~\ref{fig:arch}(b), where
the square $[-L,L]^2\times\{0\}$ is used; $\Gset_{\rm air}=\emptyset$ throughout.
Section~\ref{sec:reduced-aperture} varies $D$, $L$ and the position of $\Gset$.

\subsection{Grid Resolution Requirement}\label{sec:grid-resolution}
\begin{figure}[!t]
\centering
\includegraphics[width=0.95\columnwidth]{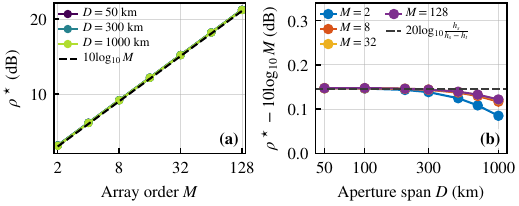}
\vspace{-10pt}
\caption{Geometric mean-SIR ceiling. (a) $\rho^\star$ versus array order $M$ and aperture $D$. (b) Residual above $10\log_{10}M$ compared with Theorem~\ref{thm:asymp}. Grid tolerance: $5\times10^{-3}$~dB.}
\label{fig:bound-vs-MD}
\end{figure}
All quantities involving the ground field ($\Rg$, $\rho^\star$, the
cutting-plane working set, and the reported SIR) are computed by sampling
over the continuous region $\Gset$. Since $\Rg$ depends on $\ba(\bq)$
through inter-platform phase differences, the required sampling step
is determined by the fringe period along the protection segment,
\begin{equation}\label{eq:fringe}
\Lambda_x \;=\; \frac{2\pi}{k\,D/h_s} \;=\; \frac{\lambda h_s}{D},
\end{equation}
which equals $0.30$~m for the parameters of
Section~\ref{sec:numerics}. Nyquist sampling over the $2L=40$~km
segment therefore requires
\begin{equation}\label{eq:Nreq}
N_{\rm g} \;>\; \frac{4LD}{\lambda h_s} \;\approx\; 2.7\times10^{5}
\end{equation}
points, three orders of magnitude more than the grids typically
used for this class of problem. This has two important consequences.

First, covariance integration must converge. The estimate of
$\rho^\star$ changes non-monotonically on coarse grids
(e.g., $8.931$~dB at $N_{\rm g}=401$), but yields
$9.1748$ and $9.1745$~dB at $25\,601$ and $4.1\times10^5$
points, respectively; independent adaptive quadrature gives
$9.1745$~dB. Subsequent mean-SIR results use the converged grid.

Second, plotted point samples alias the subpixel fringes.
Fig.~\ref{fig:protection} therefore shows cellwise peak envelopes.
A $25\,601$-point grid suggests ten SOCP threshold contacts,
whereas twice-Nyquist sampling resolves two at
$x=\pm11.12$~km.

\begin{remark}[Design grid versus evaluation grid]\label{rem:grids}
Requirement~\eqref{eq:Nreq} applies to peak evaluation, not
necessarily the initial design grid. We use $N_{\rm opt}=401$
for design and $N_{\rm eval}=25\,601$ for mean-SIR evaluation
and finite-grid separation; oracle MVDR and constant-modulus
beamformers use the converged covariance. Changing design-
and oracle-grid densities by factors of ten and forty, respectively,
changes the SOCP mean SIR by at most $0.03$~dB. Peak locations
are checked independently at twice Nyquist density.
These checks do not certify continuous-region peak feasibility:
a rigorous certificate additionally requires a global maximization
bound or an interpolation-error guarantee.
\end{remark}

% \subsection{Validation of Theorems~\ref{thm:bound} and~\ref{thm:asymp}}
% \begin{figure}[!t]
% \revfloat
% \centering
% \includegraphics[width=\columnwidth]{fig_bound_vs_M.pdf}
% \caption{\rev{(a) SIR upper bound $\rho^\star$ versus array order $M$ for aperture spans $D$ from $50$ to $1000$~km, with the $10\log_{10}M$ asymptote. The seven curves overlap almost exactly. (b) The same data plotted as $\rho^\star-10\log_{10}M$ versus $D$, which resolves that overlap. The excess is independent of $D$ up to $500$~km and equals the closed-form constant $20\log_{10}[h_s/(h_s-h_t)]$ of Theorem~\ref{thm:asymp} to within $0.01$~dB. Every point uses a grid certified converged to $5\times10^{-3}$~dB.}}
% \label{fig:bound-vs-MD}
% \end{figure}
% Fig.~\ref{fig:bound-vs-MD} shows $\rho^\star$ against $M$ for several
% aperture spans. \rev{The bound is essentially independent of $D$, the spread
% across the seven spans staying below $0.023$~dB for every $M$ once
% $D\le500$~km and below $0.063$~dB up to $1000$~km. It increases by approximately $3$~dB
% whenever $M$ doubles, confirming Theorem~\ref{thm:asymp}, and lies within
% $0.16$~dB of $10\log_{10}M$ throughout, an offset that
% Fig.~\ref{fig:bound-vs-MD}(b) resolves as the closed-form constant
% $20\log_{10}[h_s/(h_s-h_t)]=0.146$~dB rather than as scatter. The reference
% point ($M=8$, $D=300$~km) gives $\rho^\star=9.17$~dB on the converged grid of
% Section~\ref{sec:grid-resolution}.}

\subsection{Validation of Theorems~\ref{thm:bound} and~\ref{thm:asymp}}
The mean-SIR ceiling in Fig.~\ref{fig:bound-vs-MD} grows
by approximately $3$~dB per doubling of $M$, while varying by
less than $0.063$~dB over the tested aperture spans. The offset
from $10\log_{10}M$ tracks the predicted
$20\log_{10}[h_s/(h_s-h_t)]=0.146$~dB. At the reference geometry,
$\rho^\star=9.17$~dB.

\begin{figure}[!t]
\centering
\includegraphics[width=0.95\columnwidth]{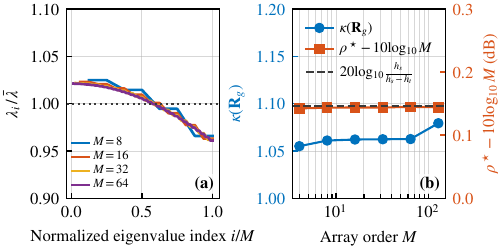}
\vspace{-10pt}
\caption{Ground covariance spectrum. (a) Normalized eigenvalues. (b) Condition number $\kappa(\Rg)$ and offset $\rho^\star-10\log_{10}M$ against the asymptotic prediction.}
\label{fig:spectral}
\end{figure}
% Fig.~\ref{fig:spectral} provides the spectral diagnostic of
% Corollary~\ref{cor:spectral}. \rev{Normalized eigenvalues remain within $\pm3\%$ of $\bar\lambda$ for $M$ from $8$ to $64$, and
% $\kappa(\Rg)=1.055$--$1.079$ over $M=4$ to $128$, essentially independent of
% $M$.} Flatness implies $\Rg \approx \bar\lambda \bI$, reducing
% $\ba_t^H\Rg^{-1}\ba_t$ to the pure array gain
% $\rho^\star \approx \|\ba_t\|^2/\bar\lambda$\rev{, which by
% Theorem~\ref{thm:asymp} equals $M[h_s/(h_s-h_t)]^2$. The measured excess over $10\log_{10}M$ is $0.143$--$0.145$~dB against the closed-form
% $0.146$~dB}.

Fig.~\ref{fig:spectral} explains this scaling: the eigenvalues
of $\Rg$ remain within $\pm3\%$ of their mean and
$\kappa(\Rg)\approx1.06$. Thus, $\Rg\approx\bar\lambda\bI$,
so $\rho^\star\approx\|\ba_t\|^2/\bar\lambda$, recovering
the array-order scaling of Theorem~\ref{thm:asymp}.

\begin{figure}[!t]
\centering
\includegraphics[width=0.95\columnwidth]{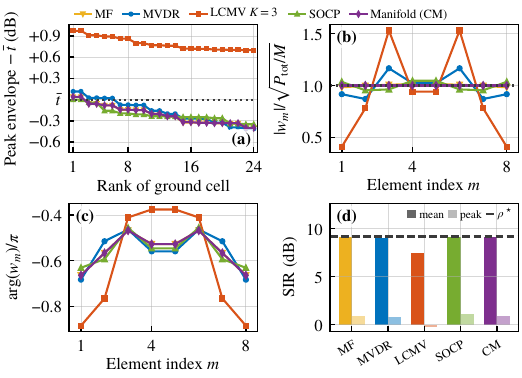}
\vspace{-10pt}
\caption{Beamformer comparison at $M=8$. (a) Ranked peak-envelope margin relative to the SOCP reference $\bar t$. (b) Weight magnitudes. (c) Weight phases. (d) Mean and sampled peak SIR.}
\label{fig:profiles}
\end{figure}

\subsection{Algorithmic Comparison}
\begin{table}[!h]
\centering\small
\setlength{\tabcolsep}{3.0pt}
% \caption{\rev{Numerical comparison at $M=8$ and $D=300$~km.
% MVDR, LCMV and SOCP use the design grid with $N_{\rm opt}=401$ points,
% while oracle MVDR and the five constant-modulus methods use the
% converged covariance. All methods are evaluated on the converged grid
% with $N_{\rm eval}=25\,601$, where $\rho^\star=9.17$~dB.
% Latencies are medians of seven runs (three for SOCP) on the same
% machine. Where two latencies are reported, they correspond to
% MRT and MVDR-projected initialization, respectively.
% The MRT initialization requires $9$ iterations, compared with a
% median of $25$ for random-phase initialization
% (Remark~\ref{rem:mvdr-init}). The peak SIR is evaluated on the same
% grid, which is converged for the mean-SIR integral but remains below
% the Nyquist density of~\eqref{eq:Nreq}. The peak values therefore
% provide a relative comparison rather than certified worst-case levels.}}

\caption{Beamforming comparison at $M=8$, $D=300$~km. Design grid: $401$ points; SIR evaluation: $25\,601$ points. Peak-SIR values are sampled-grid comparisons, not certified continuous-region maxima. Latencies are medians of seven runs (three for SOCP); paired values use MRT/MVDR starts.}
\label{tab:numerics}
\vspace{-5pt}
\small
\begin{tabular}{@{}lcccc@{}}
\toprule
Method & Latency & SIR$_{\text{mean}}$ & SIR$_{\text{peak}}$ & $\Delta\rho^\star$ \\
       &         & (dB)                & (dB)                & (dB) \\
\midrule
Matched filter            & $8~\mu$s   & $9.17$  & $0.93$  & $-0.00$ \\
MVDR (Alg.~\ref{alg:mvdr})      & $16~\mu$s  & $9.07$  & $0.82$  & $-0.10$ \\
MVDR, oracle $\Rg$        & $15~\mu$s  & $\mathbf{9.17}$ & $0.94$ & $\mathbf{+0.00}$ \\
LCMV $K{=}3$                    & $51~\mu$s  & $7.46$  & $-0.26$ & $-1.71$ \\
LCMV $K{=}5$                    & $50~\mu$s  & $4.72$  & $-3.44$ & $-4.46$ \\
SOCP (Alg.~\ref{alg:socp})      & $134$~ms   & $9.14$  & $1.08$  & $-0.03$ \\
Manifold, MRT init              & $0.96$~ms  & $9.17$  & $0.93$  & $-0.00$ \\
Manifold, MVDR init             & $0.44$~ms  & $9.17$  & $0.93$  & $-0.00$ \\
Projected gradient        & $0.75\,/\,0.44$~ms & $9.17$ & $0.93$ & $-0.00$ \\
Coordinate ascent         & $0.67\,/\,0.19$~ms & $9.17$ & $0.93$ & $-0.00$ \\
SDR + randomization       & $36$~ms    & $9.17$  & $0.93$  & $-0.00$ \\
\bottomrule
\end{tabular}
\end{table}

% Three observations support the theory:

% (O1) \rev{Given the covariance on the evaluation grid, MVDR attains
% $\rho^\star$ exactly ($\Delta\rho^\star=+0.00$~dB) as the Capon identity
% requires, and gives up $0.10$~dB when designed instead on the coarse $N_{\rm opt}=401$ covariance. This isolates the cost of estimating $\Rg$ from
% the cost of the beamformer itself.}

% (O2) \rev{Every method in Table~\ref{tab:numerics} lies at or below
% $\rho^\star$, the cutting-plane SOCP by $0.03$~dB, as
% Theorem~\ref{thm:bound} requires. The matched filter comes within
% $0.002$~dB of $\rho^\star$ without using $\Rg$, a
% consequence of the near-scalar spectrum of Fig.~\ref{fig:spectral}.}

% (O3) \rev{The constant-modulus solution of
% Algorithm~\ref{alg:manifold}, designed on the converged covariance, reaches $9.17$~dB, only
% $5\times10^{-4}$~dB below the free-amplitude bound, so the optimum is intrinsically near
% constant modulus. Designed on the coarse $N_{\rm opt}$ covariance instead, every
% constant-modulus solver stops at $9.13$~dB; this $0.04$~dB loss is due to the coarse design grid, as in (O1).
% On the converged covariance the MVDR weights vary in amplitude by at most
% $1.6\%$ ($17\%$ on the coarse one, Fig.~\ref{fig:profiles}(b)), so their projection onto
% $\Mset_{\rm CM}$ is within $10^{-12}$~dB of the constant-modulus optimum, whereas LCMV
% spreads amplitudes over a factor of $3.7$.}

Table~\ref{tab:numerics} and Fig.~\ref{fig:profiles} provide three consistency checks.
First, oracle MVDR attains $\rho^\star$, whereas design
on the coarse covariance costs $0.10$~dB. Second, SOCP lies
$0.03$~dB below the mean-SIR ceiling, and the matched filter
lies within $0.002$~dB because $\Rg$ is nearly scalar.
Third, the converged-covariance constant-modulus solution
falls only $5\times10^{-4}$~dB below the free-amplitude bound;
using the coarse covariance instead incurs $0.04$~dB.

Projected-gradient, coordinate-ascent, and SDR baselines
recover the same phase solution within $2\times10^{-6}$~rad;
the rank-one SDR certificate limits possible improvement to
$6\times10^{-11}$~dB. MVDR-projected initialization requires
one Riemannian iteration ($0.44$~ms), versus nine ($0.96$~ms)
from MRT.

\subsection{{Reduced Aperture}}\label{sec:reduced-aperture}
\begin{table}[!t]
\centering\small
\setlength{\tabcolsep}{3.2pt}
\caption{Mean SIR (dB) of the matched filter (MF) and the ceiling $\rho^\star$ versus aperture $D$ ($M=8$, $L=20$~km). Shifted segment: $10\le x\le50$~km. $\kappa(\Rg)$: centered segment.}
\label{tab:aperture}
\vspace{-5pt}
\begin{tabular}{@{}lcccccc@{}}
\toprule
 & & & \multicolumn{2}{c}{Centered} & \multicolumn{2}{c}{Shifted}\\
\cmidrule(lr){4-5}\cmidrule(l){6-7}
$D$ & $\Phi_{\min}$ & $\kappa(\Rg)$ & MF & $\rho^\star$ & MF & $\rho^\star$\\
\midrule
$300$~km & $6.0\times10^{4}$ & $1.06$ & $9.17$ & $9.17$ & $9.18$ & $9.18$\\
$100$~m & $20$ & $1.17$ & $8.80$ & $8.80$ & $9.38$ & $9.38$\\
$20$~m & $4.0$ & $1.99$ & $10.14$ & $10.17$ & $10.13$ & $10.15$\\
$16$~m & $3.2$ & $1.12$ & $9.25$ & $9.25$ & $9.24$ & $9.24$\\
$12$~m & $2.4$ & $45$ & $8.07$ & $8.22$ & $15.28$ & $21.46$\\
$10$~m & $2.0$ & $1.1\times10^{3}$ & $7.31$ & $7.82$ & $19.08$ & $36.57$\\
\bottomrule
\end{tabular}
\end{table}
The matched filter is near-optimal in Table~\ref{tab:numerics}
because $\Rg$ is near-scalar (Fig.~\ref{fig:spectral}), which
Theorem~\ref{thm:asymp} ties to a large
$\Phi_{\min}=kLd_{\min}/h_s$. Reducing $D$ lowers $\Phi_{\min}$
proportionally. In Table~\ref{tab:aperture}, $\rho^\star$ stays
within $0.03$~dB of the matched filter down to $D=16$~m, where
$2L$ still spans one fringe period $\lambda h_s/d_{\min}$ of
adjacent platforms ($\Phi_{\min}\approx\pi$). The non-monotonic
$\kappa(\Rg)$ and $\rho^\star$ for $D\ge16$~m reflect the
sinc-type off-diagonal entries of $\Rg$
(Appendix~\ref{app:prop-asymp}), which vanish for all pairs of the
equally spaced arc at $\Phi_{\min}=\pi$. For smaller $D$, $\Rg$
becomes ill-conditioned and covariance-aware designs outperform the
matched filter. Such apertures correspond to a compact cluster or a
single-platform array. For the centered segment, $\rho^\star$ exceeds the
matched filter by only $0.51$~dB at $D=10$~m, and attaining it
costs $5.9$~dB of target power, because below $D_{\rm crit}$ the target and the segment center
share one far-field direction. For the shifted segment,
Algorithm~\ref{alg:mvdr} raises the mean and peak SIR by $17.5$
and $11.1$~dB at a $1.3$~dB target-power cost, and
Algorithm~\ref{alg:socp} raises the peak SIR by $6.8$~dB at a
$0.15$~dB cost. The constant-modulus Algorithm~\ref{alg:manifold}
gains only $0.6$~dB. Since $\Phi_{\min}$ depends on $Ld_{\min}$,
narrow protected sites reach this regime at the full aperture. For
$D=300$~km and a single receiver modeled as a segment of
half-length $L=0.5$~m
($\Phi_{\min}=1.5$), Algorithm~\ref{alg:mvdr} attains
$\rho^\star=38.5$~dB at a $1.2$~dB target-power cost, versus
$12.6$~dB for the matched filter. Covariance-aware designs thus pay
off when $\Gset$ spans less than one fringe period and the
formation can resolve it from the target, whereas for
kilometer-scale baselines protecting kilometer-scale regions
$\rho^\star$ certifies that the matched filter is near-optimal.

\subsection{LCMV Failure Analysis}\label{sec:lcmv-failure}
\begin{figure}[!t]
\centering
\includegraphics[width=0.95\columnwidth]{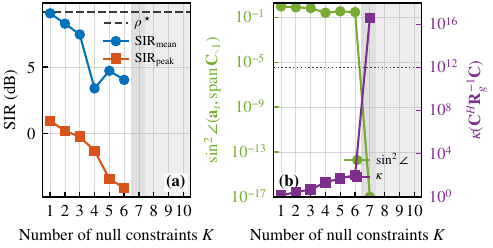}
\vspace{-10pt}
\caption{LCMV versus hard-null count $K$. (a) Mean/peak SIR for numerically admissible $K\le6$. (b) Target--null subspace angle and Gram-matrix condition number. Shading marks the nearly singular $K=7$ case and structurally singular $K\ge8$ cases.}
\label{fig:lcmv-fail}
\end{figure}
Fig.~\ref{fig:lcmv-fail} isolates the LCMV failure mechanism:
mean SIR drops from $9.06$~dB at $K=1$ to $4.72$~dB at $K=5$
as the target direction approaches the null span. For the
symmetric geometry, the Gram condition number exceeds $10^{16}$
at $K=7$; for $K\ge M=8$ it is structurally singular.
Only $K\le6$ values are reported as valid performance.
The $4.5$~dB loss already present at $K=5$ shows that agreement
among the other beamformers is not an algorithmic identity.

\subsection{Architectural Validation}
\begin{figure}[!t]
\centering
\includegraphics[width=0.95\columnwidth]{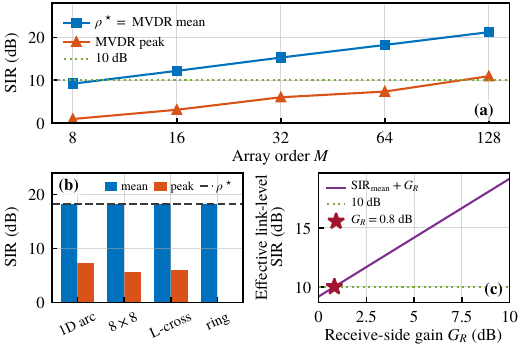}
\vspace{-10pt}
\caption{Architectural comparisons. (a) Mean-SIR ceiling and peak SIR versus $M$. (b) Four $M=64$ layouts evaluated over a two-dimensional protection region. (c) Illustrative link-level ratio with conditional receive gain $G_R$.}
\label{fig:arch}
\end{figure}

Fig.~\ref{fig:arch} compares three architectural options.
The mean SIR follows the $10\log_{10}M$ scaling, but a peak SIR
above $10$~dB requires $M=128$. For the $M=64$ layouts,
$\rho^\star$ differs by less than $0.03$~dB when evaluated over
the two-dimensional protection region. A one-dimensional cut
would make symmetric layouts rank-deficient and yield
regularization-dependent values. In contrast, peak SIR depends
on layout; the ring is geometrically degenerate at the nadir point.

Under the normalized receive-gain assumption of
Remark~\ref{rem:per-platform}, Fig.~\ref{fig:arch}(c) illustrates
that a $0.8$~dB desired-receiver gain raises the effective
link-level ratio to $10$~dB at $M=8$. These results distinguish the
$M$-weight hierarchical geometric ceiling from a universal
information-theoretic limit; joint element-level processing
requires additional network-wide coherence.

\subsection{Robustness to Position Perturbation}\label{sec:numerics-robust}

With MVDR weights fixed at the nominal geometry,
Fig.~\ref{fig:robustness} verifies the small-error quadratic loss
predicted by~\eqref{eq:perturb-loss}. At $\sigma_p=4.7$~mm,
simulation and the phase-only formula give $0.149$ and
$0.150$~dB, respectively; the model predicts less than
$0.5$~dB loss for $\sigma_p\lesssim8.7$~mm.

\begin{figure}[!t]
\centering
\includegraphics[width=0.95\columnwidth]{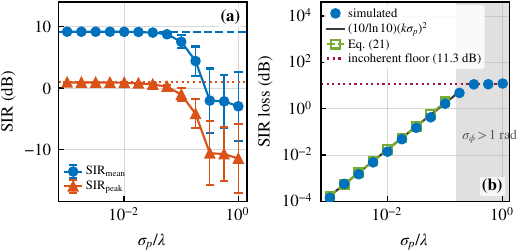}
\vspace{-10pt}
\caption{Position-error robustness with fixed nominal weights.
(a) Mean and peak SIR. (b) Decibel-averaged loss against the
Gaussian-phase approximation and~\eqref{eq:perturb-loss};
shading indicates $\sigma_\phi>1$~rad.}
\label{fig:robustness}
\end{figure}

\subsection{Per-Platform Aperture Strategies}\label{sec:per-platform-numerics}
\begin{figure}[!t]
\centering
\includegraphics[width=0.95\columnwidth]{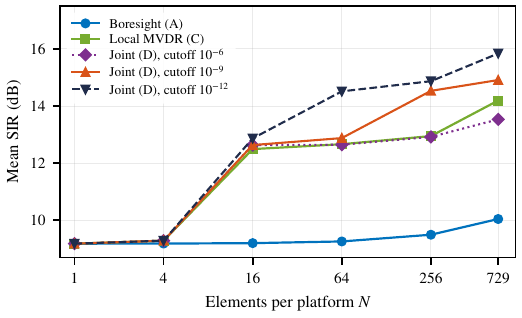}
\vspace{-10pt}
\caption{Per-platform aperture comparison ($M=8$): boresight (A), loaded local MVDR (C), and joint MVDR (D) at three spectral cutoffs. Local hard-null results (B) failing the $10^{-4}$ residual test are excluded.}
\label{fig:per-platform}
\end{figure}
We validate Remark~\ref{rem:per-platform} by comparing four
per-platform aperture strategies under the reference geometry.
Each platform carries an $N_{\rm side}\times N_{\rm side}$
planar array ($N=N_{\rm side}^2$) with $\lambda/2$ spacing,
reaching the reference $1.95~\text{m}\times1.95~\text{m}$
aperture at $N_{\rm side}=27$ and $f=2$~GHz.
Let $\ba_{t,m}\triangleq\ba_m(\bq^\star)$ denote the local
target steering vector. All strategies are compared through
scale-invariant SIRs, so each full-element beam can be normalized
to the same $P_{\rm tot}$, and the coherent local array gain is not
an extra transmit-power budget.
We consider four strategies.
\textbf{(A) Boresight steering} uses
$\bw_m^{(\rm A)} \propto \ba_{t,m}$.
\textbf{(B) Per-platform null steering} applies local LCMV
with $K=\min(N-1,5)$ hard nulls at representative points
in $\Gset$.
\textbf{(C) Per-platform MVDR} uses
$\bw_m^{(\rm C)} = \bR_{g,m}^{-1}\ba_{t,m}$,
where $\bR_{g,m}$ is the ground covariance of platform $m$,
with diagonal loading limiting its condition number to $10^6$.
\textbf{(D) Joint $MN$-dimensional MVDR} uses
$\bw^{(\rm D)} = \bR_{\rm J}^{+}\ba_{\rm J}$,
where $\ba_{\rm J}\in\Cset^{MN}$ is the target steering
vector of all $MN$ elements and $\bR_{\rm J}^{+}$ is the
pseudoinverse of their ground covariance, obtained by
discarding eigenvalues below $10^{-9}$ of the largest.

Fig.~\ref{fig:per-platform} separates the small gain of
hierarchical local apertures from the regularization-sensitive
gain of fully joint processing. Boresight (A) rises from
$9.17$ to $10.04$~dB at $N=729$. For (B), let $\bC_m=[\ba_{t,m},\ba_m(\bq_1),\ldots,
\ba_m(\bq_K)]$ and $\bef=[1,0,\ldots,0]^T$.
We omit any point for which
$\max_m\{\|\bC_m^H\bw_m-\bef\|_2,\,
\|[\bC_m^H\bw_m]_{2:K+1}\|_2/|\ba_{t,m}^H\bw_m|\}>10^{-4}$.
Loaded local MVDR (C) also depends strongly on its conditioning
cap, except for the well-defined $N=1$ scalar case.
For joint MVDR (D) at $N=729$, spectral cutoffs
$10^{-6}$, $10^{-9}$, and $10^{-12}$ give $13.53$, $14.91$,
and $15.84$~dB, respectively: these are regularized achievable
values, not cutoff-independent ceilings. A compressed
spherical-wave response matches independent field samples
within $7\times10^{-9}$ relatively; quadrupling the integration
grid changes the two finer-cutoff results by less than
$8\times10^{-4}$~dB.
The hierarchical limitation follows from the $4.4^\circ$ local
beamwidth versus the approximately $1.9^\circ$ ground-region
half-angular extent (distinct from the target--sub-nadir angular
difference of at most $0.23^\circ$). Full $MN$-element processing retains phase degrees
of freedom discarded by the $M$-weight architecture.

\section{Conclusion}\label{sec:conclusion}

This paper developed a unified Lagrangian-dual framework for sparse-array
near-field beam focusing with spatial interference suppression. The analysis
established the optimal beamformer structure, a finite atomic
representation of an optimal dual measure, and an upper bound on the
achievable mean SIR. The framework connects MVDR, LCMV and
cutting-plane SOCP through a common structure, while the mean-SIR
bound also serves as a benchmark for constant-modulus designs.

For the geometries considered, the mean-SIR bound scales linearly
with $M$. The tested constant-modulus designs approach this bound,
leaving limited room for further algorithmic improvement at the
fixed geometry. Covariance-aware designs gain substantially over the
matched filter when the protected region spans less than one fringe
period of adjacent platforms and the formation can resolve it from
the target (Section~\ref{sec:reduced-aperture}).
The ceiling applies to the two-tier
hierarchical architecture considered here. Higher cutoff-dependent joint-processing ratios may be obtained
by retaining additional physical-element modes; receiver-side
selectivity provides a separate conditional link-level gain.
Future work includes vector electromagnetic models, dynamic geometries,
joint transmit--receive optimization, and robustness to channel
uncertainty and hardware impairments.

% \appendices
% \section{Supplementary Material}\label{app:supp}

%% ===================== APPENDICES =====================
\FloatBarrier
\appendices

\section{Proofs of Theorems~\ref{thm:structure} and~\ref{thm:sparsity}}\label{app:thm-structure}

\runin{Proof of Theorem~\ref{thm:structure}.}
The feasible set of~\eqref{eq:canonical} is closed,
since $\ba(\cdot)$ is continuous, and bounded by
$\|\bw\|_2\le\sqrt{P_{\rm tot}}$, so the continuous objective attains its maximum on it.

The certificate constructed after~\eqref{eq:slater}
satisfies the Slater--Robinson condition, so strong duality holds and the dual
optimum is attained~\cite[Thm.~3.6]{ref:shapiro2009}; attainment in
$\mathcal{M}_+(\Gset)$ uses the weak-$*$ facts recorded in Section~\ref{sec:prelim}.

The Lagrangian is
\begin{equation}\label{eq:Lagr-quad}
\mathcal{L}(\bw,\mu,\nu) = \re\{\ba_t^H\bw\} - \bw^H \Reff(\mu,\nu)\bw
+ \tau\mu(\Gset) + \nu P_{\text{tot}}
\end{equation}
and $\Reff(\mu,\nu)\succeq0$ for every $(\mu,\nu)\succeq0$, so
$\mathcal{L}(\cdot,\mu,\nu)$ is concave though not necessarily strictly.
Strong duality makes any pair $(\bw^\star,(\mu^\star,\nu^\star))$ of primal
and dual optima a saddle point, so $\bw^\star$ maximizes
$\mathcal{L}(\cdot,\mu^\star,\nu^\star)$ over $\Cset^M$ and Wirtinger
stationarity~\cite{ref:hjorungnes2007} gives~\eqref{eq:wstar-kkt}. The
supremum is finite if and only if $\ba_t\in\mathrm{range}(\Reff)$, the maximizers forming
$\tfrac12\Reff^\dagger\ba_t+\ker(\Reff)$ with $\dagger$ the Moore--Penrose
pseudoinverse, whence
$g(\mu,\nu)=\tfrac14\ba_t^H\Reff^\dagger\ba_t+\tau\mu(\Gset)+\nu P_{\rm tot}$,
reducing to $\Reff^{-1}$ when $\Reff\succ0$.

If $\Reff(\mu^\star,\nu^\star)\succ0$ then
$\ker(\Reff^\star)=\{\mathbf{0}\}$ and
$\mathcal{L}(\cdot,\mu^\star,\nu^\star)$ is strictly concave, so its maximizer
is unique and equals $\tfrac12\Reff^{\star-1}\ba_t$. Every primal optimum
maximizes it, which gives~\eqref{eq:wstar}.

\runin{Proof of Theorem~\ref{thm:sparsity}.}\label{app:thm-sparsity}
\runin{Step 1 (the support lies in a compact active set).} By
complementary slackness for convex semi-infinite
programs~\cite[Thm.~5.91]{ref:bonnans2000}, every dual optimal measure
obeys $\mathrm{supp}(\mu^\star) \subseteq \Gset^\star$. Continuity of
$\ba(\cdot)$ makes $\bq \mapsto |\ba(\bq)^H\bw^\star|^2$ continuous, so
$\Gset^\star$ is closed; as a closed subset of the compact set $\Gset$
it is compact.

\runin{Step 2 (the stationarity vector lies in a $(2M-2)$-dimensional affine set).}
Let $\bw^\star$ be a primal and $(\mu^\star,\nu^\star)$ a dual optimum. By~\eqref{eq:wstar-kkt},
$\bM(\mu^\star)\bw^\star=\tfrac12\ba_t-\nu^\star\bw^\star=\int_{\Gset^\star}\bv\,d\mu^\star$, where
\begin{equation}\label{eq:testfun}
\begin{aligned}
\bv(\bq)&\triangleq\ba(\bq)\ba(\bq)^H\bw^\star,\\
\bw^{\star H}\bv(\bq)&=|\ba(\bq)^H\bw^\star|^2=\tau\quad\forall\bq\in\Gset^\star.
\end{aligned}
\end{equation}
If $\mu^\star=0$ there is nothing to prove. Otherwise $\Gset^\star\supseteq\mathrm{supp}(\mu^\star)$
is non-empty and $\tau>0$ by~\eqref{eq:slater}, so $\bw^\star\neq\mathbf 0$ and the second relation places
$\bv(\Gset^\star)$ in the real affine set $\{\bx\in\Cset^M:\bw^{\star H}\bx=\tau\}$ of
dimension $2M-2$.

\runin{Step 3 (Carath\'eodory).} The point
$\int\bv\,d\mu^\star/\mu^\star(\Gset)$ lies in the convex hull of the compact set
$\bv(\mathrm{supp}\,\mu^\star)$, and Carath\'eodory's theorem in that affine set, the
finite-dimensional case of Tchakaloff's theorem~\cite{ref:tchakaloff},~\cite[Thm.~2]{ref:bayer-teichmann}
for the $2M-2$ affine coordinates of $\bv$ and the constant function, supplies $K\le2M-1$ points
$\bq_1^\star,\dots,\bq_K^\star\in\mathrm{supp}(\mu^\star)\subseteq\Gset^\star$ and weights
$c_1,\dots,c_K>0$ with
$\sum_{i=1}^{K}c_i\bv(\bq_i^\star)=\int\bv\,d\mu^\star$
and $\sum_{i=1}^{K}c_i=\mu^\star(\Gset)$.
Thus, both the stationarity vector and the dual-measure mass are preserved.

\runin{Step 4 (the atomic measure is dual optimal).} Set
$\mu_{\rm atom}\triangleq\sum_{i=1}^{K}c_i\delta_{\bq_i^\star}\in\Mset_+(\Gset)$. Then
$\bM(\mu_{\rm atom})\bw^\star=\bM(\mu^\star)\bw^\star$, so~\eqref{eq:wstar-kkt} holds for
$(\mu_{\rm atom},\nu^\star)$ and $\bw^\star$ maximizes the concave
$\mathcal{L}(\cdot,\mu_{\rm atom},\nu^\star)$. As every $\bq_i^\star\in\Gset^\star$,
$\bw^{\star H}\bM(\mu_{\rm atom})\bw^\star=\tau\mu_{\rm atom}(\Gset)$, and
$\nu^\star(P_{\rm tot}-\|\bw^\star\|^2)=0$ by complementary slackness, so
$g(\mu_{\rm atom},\nu^\star)=\mathcal{L}(\bw^\star,\mu_{\rm atom},\nu^\star)
=\re\{\ba_t^H\bw^\star\}=d^\star$. Hence $\mu_{\rm atom}$ is an optimal dual measure,
purely atomic with at most $2M-1$ atoms, all of them in $\Gset^\star$ by Step~1.

\section{Proof of Theorem~\ref{thm:asymp}}\label{app:prop-asymp}
\textbf{(a)} For $\bq=(x_g,0,0)$ and $\bp_m=(x_m,0,h_s)$,
$r_m(\bq)=h_s\sqrt{1+u_m^2}$ with $u_m=(x_m-x_g)/h_s=\mathcal{O}(\varepsilon)$,
so $1/(r_m r_n)=h_s^{-2}(1+\mathcal{O}(\varepsilon^2))$. The Fresnel
expansion $\Delta r_{mn}=r_m-r_n\approx (x_m-x_n)(x_m+x_n-2x_g)/(2h_s)$
makes the phase $k\Delta r_{mn}$ linear in $x_g$ with slope
$-k(x_m-x_n)/h_s$ and no stationary point, so
$[\Rg]_{mn}=(2L)^{-1}\!\int_{-L}^{L} e^{jk\Delta r_{mn}}/(r_m r_n)\,dx_g$
is approximately the Fourier transform of the window $[-L,L]$ evaluated at that slope,
i.e., a sinc kernel of argument $k(x_m-x_n)L/h_s$.
The integration by parts of part~(c), which keeps the exact phase, bounds
its magnitude relative to the diagonal by $Ch_s/(kL|x_m-x_n|)$, so the slowest decay is set by the
closest pair, $|x_m-x_n|=d_{\min}$, giving~\eqref{eq:offdiag}. Only for an
equally spaced 1D arc is $d_{\min}=D/(M-1)$; for an irregular arrangement
$d_{\min}$ can be far smaller, and it is $\Phi_{\min}$, not $\Phi$, that
governs the suppression, $\varepsilon$ fixing the diagonal variation and
$\Phi_{\min}$ the off-diagonal decay. Gershgorin's theorem applied to
$\bDelta$, with diagonal $\mathcal{O}(\varepsilon^2)$ and $M-1$ off-diagonal
entries per row of order $\Phi_{\min}^{-1}$, gives
$\|\bDelta\|_{\mathrm{op}}=\mathcal{O}(\varepsilon^2+M/\Phi_{\min})$.
With $[\Rg]_{mm}=h_s^{-2}(1+\mathcal{O}(\varepsilon^2))$
this yields~\eqref{eq:Rg-expansion}.

\textbf{(b)} From the exact identity
$\Rg=h_s^{-2}(\bI+\bDelta)$ with
$\|\bDelta\|_{\mathrm{op}}\ll1$, which requires both $\varepsilon^2$ and
$M/\Phi_{\min}$ small, the Neumann series gives
$\Rg^{-1}=h_s^{2}(\bI-\bDelta)+\mathcal{O}\bigl(h_s^{2}(\varepsilon^2+M/\Phi_{\min})^2\bigr)$, so
$\rho^\star = h_s^2\|\ba_t\|^2-h_s^2\,\ba_t^H\bDelta\ba_t
+\mathcal{O}\bigl(h_s^2\|\ba_t\|^2(\varepsilon^2+M/\Phi_{\min})^2\bigr)$.
Since $|\ba_t^H\bDelta\ba_t|\le\|\bDelta\|_{\mathrm{op}}\|\ba_t\|^2$,
part~(a) gives
$\rho^\star=h_s^2\|\ba_t\|^2\bigl(1+\mathcal{O}(\varepsilon^2+M/\Phi_{\min})\bigr)$.
Here $\|\ba_t\|^2=\sum_m\bigl(x_m^2+(h_s-h_t)^2\bigr)^{-1}
=\frac{M}{(h_s-h_t)^2}\bigl(1+\mathcal{O}(\varepsilon^2)\bigr)$, which
yields~\eqref{eq:asymp}. Retaining the factor $(h_s/(h_s-h_t))^2$ rather than absorbing it into the error term matters. Its deviation from one, $2h_t/h_s+\mathcal{O}((h_t/h_s)^2)$, is first order in $\varepsilon$. At the operating point of Section~\ref{sec:numerics} it accounts for $0.146$~dB, whereas the residual after its removal is $-0.0024$~dB. There $M/\Phi_{\min}=1.3\times10^{-4}$, so the Fresnel contribution is negligible.

\textbf{(c)} Let $\bD_g=\operatorname{diag}(\Rg)$,
$\bB=\bD_g^{-1/2}(\Rg-\bD_g)\bD_g^{-1/2}$ and $\bu=\bD_g^{-1/2}\ba_t$. Then
$\rho^\star=\bu^H(\bI+\bB)^{-1}\bu$ and the eigenvalues of $\bI+\bB$ lie in
$[1-\|\bB\|,1+\|\bB\|]$, so
$\rho_d/(1+\|\bB\|)\le\rho^\star\le\rho_d/(1-\|\bB\|)$ with
$\rho_d=\sum_m|[\ba_t]_m|^2/[\Rg]_{mm}$.
For the diagonal, let $s_m=x_m/h_s$, $f(u)=(1+u^2)^{-1}$ and $v$ be uniform on
$[-\ell,\ell]$, so that $h_s^2[\Rg]_{mm}=\mathbb{E}f(s_m+v)$. Taylor's theorem
with $\mathbb{E}v=0$, $\mathbb{E}v^2=\ell^2/3$ and $-2\le f''\le\frac12$ gives
$f(s_m)-\ell^2/3\le h_s^2[\Rg]_{mm}\le f(s_m)+\ell^2/12$, and multiplying by
$1+s_m^2\le1+\bar s^2$ yields
$(1+s_m^2)h_s^2[\Rg]_{mm}\in[1-4q,\,1+q]$, while
$(h_s-h_t)^2|[\ba_t]_m|^2=(1-p_m)/(1+s_m^2)$ with
$0\le p_m=s_m^2t(2-t)/[(1-t)^2+s_m^2]\le p$. Hence
$\rho_d(h_s-h_t)^2/(Mh_s^2)=M^{-1}\sum_m(1-p_m)/[(1+s_m^2)h_s^2[\Rg]_{mm}]$
lies in $[(1-p)/(1+q),\,1/(1-4q)]$.
For the off-diagonal, fix $m\neq n$, let $\Delta=|x_m-x_n|$,
$w=\ell+\bar s$, $\gamma=(1+w^2)^{-3/2}$ and $\sigma=w(1+w^2)^{-1/2}$, and
write $[\Rg]_{mn}=(2L)^{-1}\int_{-L}^{L}e^{jk\phi}g\,dx$ with
$\phi=r_m-r_n$ and $g=(r_mr_n)^{-1}\le h_s^{-2}$. Since
$\phi'(x)=\int_{x_m}^{x_n}h_s^2r(x,\xi)^{-3}d\xi$ with
$r(x,\xi)=[(x-\xi)^2+h_s^2]^{1/2}\le h_s(1+w^2)^{1/2}$, one has
$|\phi'|\ge\gamma\Delta/h_s$, and likewise $|\phi''|\le3\sigma\Delta/h_s^2$
and $|g'|\le2\sigma/h_s^3$. One integration by parts,
$\int e^{jk\phi}g\,dx=[e^{jk\phi}g/(jk\phi')]_{-L}^{L}-(jk)^{-1}\int e^{jk\phi}(g/\phi')'dx$,
then gives
$h_s^2|[\Rg]_{mn}|\le[\gamma^{-1}+\ell\sigma(2\gamma^{-1}+3\gamma^{-2})]/\Phi_{mn}$
with $\Phi_{mn}=kL\Delta/h_s$, and $[\Rg]_{mm}\ge h_s^{-2}(1+w^2)^{-1}$ turns
this into $|[\bB]_{mn}|\le C/\Phi_{mn}$ with $C$ as in Theorem~\ref{thm:asymp}(c).
The $j$th nearest element on either side of $x_m$ lies at least $jd_{\min}$
away, so Gershgorin's theorem gives
$\|\bB\|\le\max_m\sum_{n\neq m}|[\bB]_{mn}|\le CS_M/\Phi_{\min}=\bar\beta$.
Combining the three bounds proves~(c). At the operating point $C=1.27$,
$S_8=3.92$ and $\Phi_{\min}=6.0\times10^{4}$, and dense-grid evaluations of
$\Rg$ respect both bounds entry by entry.

\section{Proof of Theorem~\ref{thm:cp-conv}}\label{app:prop-cpconv}
(a) Since $\Gset_k\subseteq\Gset_{k+1}\subseteq\Gset$, each outer
iteration solves a tighter relaxation of~\eqref{eq:socp}; the maximum
therefore cannot increase, and $J_k\ge J^\star$ for every $k$.

(b) The iterates are bounded: $\|\bw_k\|_2\le\sqrt{P_{\rm tot}}$
and, since $\beta>0$, each finite-grid optimum satisfies
$t_{g,k}=\max_{\bq\in\Gset_k}|\ba(\bq)^H\bw_k|$, whence
$0\le t_{g,k}\le\sqrt{P_{\rm tot}}\max_{\bq\in\Gset}\|\ba(\bq)\|$.
Define the exact-oracle violation
$r_k=\max_{\bq\in\Gset}|\ba(\bq)^H\bw_k|-t_{g,k}\ge0$.
Suppose $r_k\not\to0$. There exist $\epsilon_0>0$ and indices
$k_i\to\infty$ with $r_{k_i}\ge\epsilon_0$. By compactness,
a further subsequence (relabeled) satisfies
$(\bw_{k_i},t_{g,k_i})\to(\bar\bw,\bar t)$ and
$\bq_{k_i}\to\bar\bq$, where $\bq_{k_i}$ is an exact
maximum-violation point. Continuity then implies
$|\ba(\bar\bq)^H\bar\bw|-\bar t\ge\epsilon_0$.
Yet each cut $\bq_{k_i}$, appended at iteration $k_i$ by
Algorithm~\ref{alg:socp} and never removed, remains in every later working set, so
$|\ba(\bq_{k_i})^H\bw_{k_j}|\le t_{g,k_j}$ for all $j>i$.
First let $j\to\infty$ for each fixed $i$, and then let
$i\to\infty$. Continuity gives
$|\ba(\bar\bq)^H\bar\bw|\le\bar t$, a contradiction.
Thus, $r_k\to0$. Every accumulation point is feasible for the
full semi-infinite program. Since $J_k$ is monotone and
$J_k\ge J^\star$, its limit exists; evaluating the objective
at a feasible accumulation point shows this limit is at most
$J^\star$, proving $J_k\downarrow J^\star$.

(c) From $r_k\to0$, for every $\varepsilon>0$ there is a finite
$K_\varepsilon$ with $r_k\le\varepsilon$ for all
$k\ge K_\varepsilon$. The same argument applies to the finite
oracle with $\Gset$ replaced throughout by
$\Gset_0\cup\Gset_{\rm eval}$; it then establishes convergence
only for that finite problem.

\section{Derivation of~\texorpdfstring{\eqref{eq:perturb-loss}}{(21)}}\label{app:perturb}
A first-order Taylor expansion of $\hat r_m(\bq)$ about $r_m$ gives the
phase-only steering perturbation $[\hat\ba(\bq)]_m \approx [\ba(\bq)]_m
e^{-jk\langle\boldsymbol{\xi}_m,\hat{\boldsymbol{r}}_m\rangle}$ with
$\hat{\boldsymbol{r}}_m \triangleq (\bq-\bp_m)/r_m$, the amplitude correction
being $\mathcal{O}(\|\boldsymbol{\xi}\|/r_m)\sim10^{-7}$ at LEO ranges. The
projection of $\boldsymbol{\xi}_m$ along $\hat{\boldsymbol{r}}_m$ is
$\mathcal{N}(0,\sigma_p^2)$, so
$\E[e^{-jk\langle\boldsymbol{\xi}_m,\hat{\boldsymbol{r}}_m\rangle}]
= e^{-\sigma_\phi^2/2}$, and Cox's identity~\cite{ref:cox1987},
$\E[\hat{\mathbf{X}}] = e^{-\sigma_\phi^2}\mathbf{X}+(1-e^{-\sigma_\phi^2})\mathrm{diag}(\mathbf{X})$,
applied to $\mathbf{X} = \ba_t\ba_t^H$ and, pointwise in $\bq$ and then averaged over $\Gset$, to $\mathbf{X} = \Rg$,
gives the numerator and denominator of~\eqref{eq:perturb-loss}; Section~\ref{sec:numerics-robust} compares them
with simulation.
\bibliographystyle{IEEEtran}
\bibliography{references}

@article{ref:nf-survey,
  author  = {Liu, Yuanwei and Wang, Zhaolin and Xu, Jiaqi and Ouyang, Chongjun
             and Mu, Xidong and Schober, Robert},
  title   = {Near-field communications: A tutorial review},
  journal = {IEEE Open J. Commun. Soc.},
  volume  = {4},
  pages   = {1999--2049},
  year    = {2023},
}

@article{ref:dmimo,
  author  = {Lozano, Angel and Heath, Robert W. and Andrews, Jeffrey G.},
  title   = {Fundamental limits of cooperation},
  journal = {IEEE Trans. Inf. Theory},
  volume  = {59},
  number  = {9},
  pages   = {5213--5226},
  month   = sep,
  year    = {2013},
}

@article{ref:formation,
  author  = {Nanzer, Jeffrey A. and Mghabghab, Serge R. and Ellison,
             Sean M. and Schlegel, Anton},
  title   = {Distributed phased arrays: Challenges and recent advances},
  journal = {IEEE Trans. Microw. Theory Techn.},
  volume  = {69},
  number  = {11},
  pages   = {4893--4907},
  month   = nov,
  year    = {2021},
  doi     = {10.1109/TMTT.2021.3092401},
}

@article{ref:kelley1960,
  author  = {{Kelley Jr.}, J. E.},
  title   = {The cutting-plane method for solving convex programs},
  journal = {J. Soc. Ind. Appl. Math.},
  volume  = {8},
  number  = {4},
  pages   = {703--712},
  month   = dec,
  year    = {1960},
}

@article{ref:hettich1993,
  author  = {Hettich, R. and Kortanek, K. O.},
  title   = {Semi-infinite programming: Theory, methods, and applications},
  journal = {SIAM Rev.},
  volume  = {35},
  number  = {3},
  pages   = {380--429},
  month   = sep,
  year    = {1993},
}

@article{ref:shapiro2009,
  author  = {Shapiro, A.},
  title   = {Semi-infinite programming, duality, discretization and
             optimality conditions},
  journal = {Optimization},
  volume  = {58},
  number  = {2},
  pages   = {133--161},
  month   = feb,
  year    = {2009},
}

@book{ref:bonnans2000,
  author    = {Bonnans, J. F. and Shapiro, A.},
  title     = {Perturbation Analysis of Optimization Problems},
  publisher = {Springer},
  address   = {New York, NY, USA},
  year      = {2000},
}

@article{ref:capon1969,
  author  = {Capon, J.},
  title   = {High-resolution frequency-wavenumber spectrum analysis},
  journal = {Proc. IEEE},
  volume  = {57},
  number  = {8},
  pages   = {1408--1418},
  month   = aug,
  year    = {1969},
}

@article{ref:frost1972,
  author  = {{Frost III}, Otis L.},
  title   = {An algorithm for linearly constrained adaptive array processing},
  journal = {Proc. IEEE},
  volume  = {60},
  number  = {8},
  pages   = {926--935},
  month   = aug,
  year    = {1972},
}

@article{ref:cox1987,
  author  = {Cox, H. and Zeskind, R. M. and Owen, M. M.},
  title   = {Robust adaptive beamforming},
  journal = {IEEE Trans. Acoust., Speech, Signal Process.},
  volume  = {ASSP-35},
  number  = {10},
  pages   = {1365--1376},
  month   = oct,
  year    = {1987},
}

@article{ref:vorobyov2003,
  author  = {Vorobyov, S. A. and Gershman, A. B. and Luo, Z.-Q.},
  title   = {Robust adaptive beamforming using worst-case performance
             optimization: A solution to the signal mismatch problem},
  journal = {IEEE Trans. Signal Process.},
  volume  = {51},
  number  = {2},
  pages   = {313--324},
  month   = feb,
  year    = {2003},
}

@book{ref:absil2008,
  author    = {Absil, P.-A. and Mahony, R. and Sepulchre, R.},
  title     = {Optimization Algorithms on Matrix Manifolds},
  publisher = {Princeton Univ. Press},
  address   = {Princeton, NJ, USA},
  year      = {2008},
}

@article{ref:sun2017,
  author  = {Sun, J. and Qu, Q. and Wright, J.},
  title   = {Complete dictionary recovery over the sphere {I}: Overview and
             the geometric picture},
  journal = {IEEE Trans. Inf. Theory},
  volume  = {63},
  number  = {2},
  pages   = {853--884},
  month   = feb,
  year    = {2017},
}

@article{ref:hjorungnes2007,
  author  = {Hj{\o}rungnes, A. and Gesbert, D.},
  title   = {Complex-valued matrix differentiation: Techniques and key
             results},
  journal = {IEEE Trans. Signal Process.},
  volume  = {55},
  number  = {6},
  pages   = {2740--2746},
  month   = jun,
  year    = {2007},
}

@article{ref:tchakaloff,
  author  = {Tchakaloff, V.},
  title   = {Formules de cubatures m\'ecaniques \`a coefficients non
             n\'egatifs},
  journal = {Bull. Sci. Math.},
  volume  = {81},
  pages   = {123--134},
  year    = {1957},
}

@article{ref:bayer-teichmann,
  author  = {Bayer, C. and Teichmann, J.},
  title   = {The proof of {Tchakaloff}'s theorem},
  journal = {Proc. Amer. Math. Soc.},
  volume  = {134},
  number  = {10},
  pages   = {3035--3040},
  month   = oct,
  year    = {2006},
}

@book{ref:carath-toeplitz,
  author    = {Grenander, U. and Szeg{\H{o}}, G.},
  title     = {Toeplitz Forms and Their Applications},
  publisher = {Univ. California Press},
  address   = {Berkeley, CA, USA},
  year      = {1958},
}

@misc{ref:you2024nfbm,
      title={Near-Field Beam Management for Extremely Large-Scale Array Communications}, 
      author={Changsheng You and Yunpu Zhang and Chenyu Wu and Yong Zeng and Beixiong Zheng and Li Chen and Linglong Dai and A. Lee Swindlehurst},
      year={2023},
      eprint={2306.16206},
      note={arXiv:2306.16206},
      archivePrefix={arXiv},
      primaryClass={eess.SP},
}

@article{ref:wang2024beamfocus,
  author  = {Wang, Z. and Mu, X. and Liu, Y.},
  title   = {Beamfocusing optimization for near-field wideband multi-user
             communications},
  journal = {IEEE Trans. Commun.},
volume={73}, number={1}, pages={555--572}, month=jan, year={2025}
}

@article{ref:li2022modular,
  author  = {Li, X. and Lu, H. and Zeng, Y. and Jin, S. and Zhang, R.},
  title   = {Near-field modeling and performance analysis of modular
             extremely large-scale array communications},
  journal = {IEEE Commun. Lett.},
  volume  = {26},
  number  = {7},
  pages   = {1529--1533},
  month   = jul,
  year    = {2022},
}

@article{ref:vazquez2016,
  author  = {V\'azquez, M. \'A. and P\'erez-Neira, A. I. and Christopoulos, D.
             and Chatzinotas, S. and Ottersten, B. and Arapoglou, P.-D.
             and Ginesi, A. and Taricco, G.},
  title   = {Precoding in multibeam satellite communications: Present and
             future challenges},
  journal = {IEEE Wireless Commun.},
  volume  = {23},
  number  = {6},
  pages   = {88--95},
  month   = dec,
  year    = {2016},
}

@article{ref:joroughi2017,
  author  = {Joroughi, V. and V\'azquez, M. \'A. and P\'erez-Neira, A. I.},
  title   = {Generalized multicast multibeam precoding for satellite
             communications},
  journal = {IEEE Trans. Wireless Commun.},
  volume  = {16},
  number  = {2},
  pages   = {952--966},
  month   = feb,
  year    = {2017},
}

@article{ref:perez2019,
  author={P\'erez-Neira, A. I. and V\'azquez, M. \'A. and Shankar, M. R. B. and Maleki, S. and Chatzinotas, S.},
  journal={IEEE Signal Process. Mag.}, 
  title={Signal Processing for High-Throughput Satellites: Challenges in New Interference-Limited Scenarios}, 
  year={2019},
  volume={36},
  number={4},
  pages={112--131},
  month=jul,
  doi={10.1109/MSP.2019.2894391}}

@article{ref:you2022hybrid,
  author  = {You, L. and Qiang, X. and Li, K.-X. and Tsinos, C. G. and
             Wang, W. and Gao, X. and Ottersten, B.},
  title   = {Hybrid analog/digital precoding for downlink massive {MIMO}
             {LEO} satellite communications},
  journal = {IEEE Trans. Wireless Commun.},
  volume  = {21},
  number  = {8},
  pages   = {5962--5976},
  month   = aug,
  year    = {2022},
}

@article{ref:you2022twin,
  author  = {You, L. and Qiang, X. and Li, K.-X. and Tsinos, C. G. and
             Wang, W. and Gao, X. and Ottersten, B.},
  title   = {Massive {MIMO} hybrid precoding for {LEO} satellite
             communications with twin-resolution phase shifters and nonlinear
             power amplifiers},
  journal = {IEEE Trans. Commun.},
  volume  = {70},
  number  = {8},
  pages   = {5543--5557},
  month   = aug,
  year    = {2022},
}

@article{ref:chen2024survey,
  author  = {Chen, R. and Long, W.-X. and Wang, B. and He, Y. and
             Sun, R. and Cheng, N. and Zheng, G. and Niyato, D.},
  title   = {Multibeam high throughput satellite: Hardware foundation,
             resource allocation, and precoding},
  journal = {IEEE Commun. Surveys Tuts.},
  volume  = {28},
  pages   = {5379--5415},
  year    = {2026},
  doi     = {10.1109/COMST.2026.3673850},
}

@article{ref:abdelsadek2022dmimo,
  author  = {Abdelsadek, M. Y. and Kurt, G. K. and Yanikomeroglu, H.},
  title   = {Distributed massive {MIMO} for {LEO} satellite networks},
  journal = {IEEE Open J. Commun. Soc.},
  volume  = {3},
  pages   = {2162--2177},
  year    = {2022},
}

@inproceedings{ref:abdelsadek2021cf,
  author    = {Abdelsadek, M. Y. and Yanikomeroglu, H. and Kurt, G. K.},
  title     = {Future ultra-dense {LEO} satellite networks: A cell-free
               massive {MIMO} approach},
  booktitle = {Proc. IEEE Int. Conf. Commun. Workshops (ICC Workshops)},
  month     = jun,
  year      = {2021},
  pages     = {1--6},
}

@inproceedings{ref:humadi2024dynamic,
  author    = {Humadi, K. and Kurt, G. K. and Yanikomeroglu, H.},
  title     = {Distributed massive {MIMO} system with dynamic clustering in
               {LEO} satellite networks},
  booktitle = {Proc. 6th Int. Conf. Commun., Signal Process., Appl. (ICCSPA)},
  month     = jul,
  year      = {2024},
  pages     = {1--6},
  doi       = {10.1109/ICCSPA61559.2024.10794349},
}

@inproceedings{ref:storek2020testbed,
  author    = {Storek, K.-U. and Schwarz, R. T. and Knopp, A.},
  title     = {Multi-satellite multi-user {MIMO} precoding: Testbed and
               field trial},
  booktitle = {Proc. IEEE Int. Conf. Commun. (ICC)},
  month     = jun,
  year      = {2020},
  pages     = {1--7},
}

@inproceedings{ref:merlano2024swarm,
  author    = {{Merlano Duncan}, J. C. and Ha, V. N. and Krivochiza, J. and
               Palisetty, R. and Eappen, G. and Vasquez, J. A. and
               Martins, W. A. and Chatzinotas, S. and Ottersten, B.},
  title     = {Harnessing the power of swarm satellite networks with
               wideband distributed beamforming},
  booktitle = {Proc. IEEE 34th Annu. Int. Symp. Pers., Indoor Mobile Radio
               Commun. (PIMRC)},
  month     = sep,
  year      = {2023},
  pages     = {1--6},
  doi       = {10.1109/PIMRC56721.2023.10294061},
}

@techreport{ref:3gpp-ntn-r19,
  author      = {{3GPP}},
  title       = {Solutions for {NR} to support non-terrestrial networks
                 ({NTN})},
    institution = {3rd Generation Partnership Project (3GPP)},
  address     = {Sophia Antipolis, France},
  type        = {Tech. Rep.},
  number      = {TR~38.821, V16.2.0},
  month       = mar,
  year        = {2023},
}

@article{ref:lu2022xlarray,author={Lu, H. and Zeng, Y.},title={Communicating with extremely large-scale array/surface: Unified modeling and performance analysis},journal={IEEE Trans. Wireless Commun.},volume={21},number={6},pages={4039--4053},month=jun,year={2022}}

@ARTICLE{ref:wu2023ldma,
  author  = {Wu, Zidong and Dai, Linglong},
  journal = {IEEE J. Sel. Areas Commun.},
  title   = {Multiple Access for Near-Field Communications: {SDMA} or {LDMA}?},
  year    = {2023},
  volume  = {41},
  number  = {6},
    pages   = {1918--1935},
  month   = jun,
  doi     = {10.1109/JSAC.2023.3275616}
}

@ARTICLE{ref:zhang2022beamfocus,
  author  = {Zhang, Haiyang and Shlezinger, Nir and Guidi, Francesco and Dardari, Davide and Imani, Mohammadreza F. and Eldar, Yonina C.},
  journal = {IEEE Trans. Wireless Commun.},
  title   = {Beam Focusing for Near-Field Multiuser {MIMO} Communications},
  year    = {2022},
  volume  = {21},
  number  = {9},
    pages   = {7476--7490},
  month   = sep,
  doi     = {10.1109/TWC.2022.3158894}
}

@article{ref:allende2018baseline,
  author  = {Allende-Alba, Gerardo and Montenbruck, Oliver and Hackel, Stefan
             and Tossaint, Michel},
  title   = {Relative positioning of formation-flying spacecraft using
             single-receiver {GPS} carrier phase ambiguity fixing},
  journal = {GPS Solut.},
  volume  = {22},
  number  = {3},
    month   = jul,
  year    = {2018},
  note    = {{A}rt. no. 68},
  doi     = {10.1007/s10291-018-0734-x}
}

@misc{ref:uso2013,
  author       = {{Jet Propulsion Laboratory}},
  title        = {Measurements of ultra-stable oscillator {(USO)} {Allan}
                  deviations in space},
  howpublished = {\emph{NASA Tech Briefs}, NPO-48705},
  month        = feb,
  year         = {2013},
}

@misc{ref:fcc2148,
  author       = {{Federal Communications Commission}},
  title        = {Space {Exploration} {Holdings}, {LLC}, request for
                  modification of the authorization for the {SpaceX} {NGSO}
                  satellite system},
  howpublished = {Order and Authorization and Order on Reconsideration,
                  FCC~21-48, 36~FCC~Rcd~7995},
  month        = apr,
  year         = {2021},
}

\end{document}